\documentclass[11pt]{article}

\usepackage[margin=1in]{geometry} 

\usepackage{subfiles}
\usepackage{xcolor}
\usepackage[colorlinks=true,citecolor=blue,linkcolor=blue]{hyperref} 
\hypersetup{
    colorlinks,
    linkcolor={red!50!black},
    citecolor={blue!50!black},
    urlcolor={blue!80!black}
}

\usepackage[]{amsmath,amssymb,amsfonts,latexsym,amsthm,enumerate,fullpage,xcolor,cite,bbm}
\usepackage[nameinlink, noabbrev, capitalize]{cleveref}
\crefname{prop}{Proposition}{Propositions}
\crefname{ineq}{inequality}{inequalities}
\creflabelformat{ineq}{#2(#1)#3}

\newtheorem{theorem}{Theorem}[section] 
\newtheorem{lemma}[theorem]{Lemma}

\newtheorem{fact}[theorem]{Fact}

\newtheorem{problem}{Problem}
\newtheorem{corollary}[theorem]{Corollary}
\newtheorem{definition}[theorem]{Definition}

\newtheorem{remark}[theorem]{Remark}
\newtheorem{THM}{Theorem}
\crefname{THM}{Theorem}{Theorems}

\usepackage{subcaption}
\usepackage{graphicx}
\usepackage{tikz}
\usetikzlibrary{positioning}
\usetikzlibrary{decorations.pathmorphing}
\usetikzlibrary{automata,positioning}
\usetikzlibrary{arrows}
\usetikzlibrary{arrows.meta}

\newcommand{\E}{\mathbb{E}}
\newcommand{\N}{\mathbb{N}}

\newcommand{\F}{\mathbb{F}}

\newcommand{\poly}{\operatorname{poly}}
\newcommand{\polylog}{\operatorname{polylog}}
\newcommand{\supp}{\text{support}}
\newcommand{\pr}{{\prime}}

\newcommand{\U}{\mathbf{U}}
\newcommand{\X}{\mathbf{X}}
\renewcommand{\L}{\mathbf{L}}
\newcommand{\W}{\mathbf{W}}
\newcommand{\Y}{\mathbf{Y}}

\newcommand{\zo}{\{0,1 \}}

\newcommand{\Ext}{\mathsf{Ext}}

\newcommand{\LRE}{\mathsf{LRE}}
\newcommand{\Leak}{\mathsf{Leak}}

\usepackage{algorithm,algpseudocode}

\algblockdefx[Loop]{Loop}{EndLoop}[1][]{\textbf{Loop} #1}{\textbf{End Loop}}

\algrenewcomment[1]{\(\triangleright\) #1}
\algnewcommand{\LineComment}[1]{\State \(\triangleright\) #1}
\usepackage{tabto} 
\begin{document}
\title{Improved Extractors for Small-Space Sources}
\author{Eshan Chattopadhyay\thanks{Supported by NSF CAREER Award 2045576.}\\
Cornell University\\
\texttt{eshanc@cornell.edu}
\and
Jesse Goodman\footnotemark[1]\\
Cornell University\\
\texttt{jpmgoodman@cs.cornell.edu}
}
\pagenumbering{Alph} 
\begin{titlepage}
\maketitle
\begin{abstract}
We study the problem of extracting random bits from weak sources that are sampled by algorithms with limited memory. This  model of \emph{small-space sources} was introduced by Kamp, Rao, Vadhan and Zuckerman (STOC'06), and falls into a line of research initiated by Trevisan and Vadhan (FOCS'00) on extracting randomness from weak sources that are sampled by computationally bounded algorithms. Our main results are the following.

\begin{enumerate}
    \item We obtain near-optimal extractors for small-space sources in the polynomial error regime. For space \(s\) sources over \(n\) bits, our extractors require just \(k\geq s\cdot\polylog(n)\) entropy. This is an exponential improvement over the previous best result, which required \(k\geq s^{1.1}\cdot2^{\log^{0.51} n}\) (Chattopadhyay and Li, STOC'16).
    
    \item We obtain improved extractors for small-space sources in the negligible error regime. For space \(s\) sources over \(n\) bits, our extractors require entropy \(k\geq n^{1/2+\delta}\cdot s^{1/2-\delta}\), whereas the previous best result required \(k\geq n^{2/3+\delta}\cdot s^{1/3-\delta}\) (Chattopadhyay, Goodman, Goyal and Li, STOC'20).
    
\end{enumerate}

To obtain our first result, the key ingredient is a new reduction from small-space sources to affine sources, allowing us to simply apply a good affine extractor.

To obtain our second result, we must develop some new machinery, since we do not have low-error affine extractors that work for low entropy. Our main tool is a significantly improved extractor for adversarial sources, which is built via a simple framework that makes novel use of a certain kind of leakage-resilient extractors (known as \emph{cylinder intersection extractors}), by combining them with a general type of extremal designs. Our key ingredient is the first derandomization of these designs, which we obtain using new connections to coding theory and additive combinatorics.

%
%
%
%
%
%
%
%
%
%
%
%

\end{abstract}
\thispagestyle{empty}
\end{titlepage}
\pagenumbering{arabic}

\newpage

\section{Introduction}
Randomness is a powerful computational resource that has found beautiful applications in algorithm design, cryptography, and combinatorics (see \cite{vadhan2012pseudorandomness} for an excellent survey). Unfortunately, such applications require access to uniform bits, but randomness harvested from natural phenomena (e.g., radioactive decay, atmospheric noise) rarely looks so pure. Such motivates the study of \emph{randomness extractors}, which are algorithms that convert these weak sources of randomness into distributions that are close to uniform:

\begin{definition}[Randomness extractor] 
Let \(\mathcal{X}\) be a family of distributions over \(\zo^n\). A function \(\Ext:\zo^n\to\zo^m\) is an \emph{extractor} for \(\mathcal{X}\) with error \(\epsilon\) if for every \(\X\in\mathcal{X}\),
\[
|\Ext(\X)-\U_m|\leq\epsilon,
\]
where \(\U_m\) is the uniform distribution over \(\zo^m\), and \(|\cdot|\) denotes statistical distance.
\end{definition}

Beyond purifying natural sources of randomness, extractors have found deep connections to complexity theory, cryptography, coding theory, and combinatorics (see, e.g., \cite{shaltiel2011introduction,vadhan2012pseudorandomness}). Constructing these objects has thus produced a fruitful line of research over the past 30 years, where various distribution families \(\mathcal{X}\) and errors \(\epsilon\) have been considered depending on the motivating application.

In order for extraction to be possible, each source \(\X\in\mathcal{X}\) must have \emph{some} randomness. In this field, it is standard to measure the randomness content of \(\X\) as its \emph{min-entropy}, defined as $H_\infty(\X):=\min_{x}\log(1/\Pr[\X=x])$. Unfortunately, it turns out that a min-entropy requirement alone is not enough to enable extraction. Indeed, an easy folklore argument shows that even if every source \(\X\in\mathcal{X}\) has min-entropy \(k\geq n-1\), there cannot exist an extractor \(\Ext\) for \(\mathcal{X}\) that achieves nontrivial error \(\epsilon<1/2\).

To circumvent this impossibility result, researchers have considered two main directions. In the first direction, one assumes that each source \(\X\in\mathcal{X}\) comes with a uniform seed \(\U_d\), which can be used to extract uniform bits from the rest of the source, which has some min-entropy guarantee. Extractors in this setting are called \emph{seeded extractors}, and near-optimal constructions of these objects are now known \cite{lu2003extractors,guruswami2009unbalanced,dvir2013extensions}. In this paper, we focus on the second direction, where one assumes each source \(\X\in\mathcal{X}\) has some additional structure beyond its min-entropy guarantee.


\paragraph{Samplable sources}
One natural way to equip each distribution \(\X\in\mathcal{X}\) with some additional structure is to assume that it can be \emph{sampled efficiently}, i.e., generated by an algorithm that has limited computational resources. Such sources were introduced by Trevisan and Vadhan \cite{trevisan2000extracting}, under the suggestion that they are a good model for distributions that would actually arise in nature. In \cite{trevisan2000extracting}, and the follow-up works of Viola \cite{viola2014extractors} and Li \cite{li2016improved}, the authors consider \emph{circuit sources}: distributions that can be sampled by small circuits. Such sources can be thought of as distributions sampled by algorithms with limited \emph{time}.

In this paper, we consider distributions that can be sampled by algorithms with limited \emph{memory}. Known as \emph{small-space sources}, this family of distributions was introduced by Kamp, Rao, Vadhan, and Zuckerman \cite{kamp2006small}, and further studied in recent work \cite{chattopadhyay2016extractors,adversarial-sources}. To define this class of sources formally, one uses \emph{branching programs} to model the evolution of state in the small-space algorithm. A branching program of width \(w\) and length \(n\) is a directed acyclic graph with \(n+1\) layers, where the first layer has one node, the remaining layers have \(w\) nodes each, and every edge starting in layer \(i\) terminates in layer \(i+1\). Small-space sources are then defined as follows.

\begin{definition}[Small-space source]\label{def:small-space-source}
A distribution \(\X\) over \(\zo^n\) is a \emph{space \(s\) source} if it is generated by a random walk starting on the first layer of a branching program of width \(2^s\) and length \(n\), where each edge is labeled with an output bit and some transition probability.
\end{definition}

Beyond their motivation in modeling distributions that one might actually find in nature, small-space sources are powerful enough to capture several other well-studied models. As noted in \cite{kamp2006small}, small-space sources can simulate: von Neumann's model of a coin with unknown bias \cite{von195113}; the finite Markov chain model of Blum \cite{blum1986independent}; the space-bounded models of Vazirani \cite{vazirani1987efficiency} and Koenig and Maurer \cite{koenig2004extracting,koenig2005generalized}; and the popular models of oblivious bit-fixing and symbol-fixing sources \cite{chor1985bit,kamp2006deterministic} and independent sources \cite{chor1988unbiased}. In fact, it is suggested in \cite{kamp2006small} that the only model of sources that appears unrelated to small-space sources is the class of \emph{affine sources} \cite{gabizon2008deterministic}.

\subsection{Summary of our results}
In this paper, we explicitly construct two significantly improved extractors for small-space sources. Along the way, we prove a new structural result for small-space sources, and provide new explicit constructions of several related pseudorandom objects. Our extractors follow easily from these new key ingredients, which may be of independent interest. We formally state these results, below.

\subsubsection{Small-space extractors for polylogarithmic entropy}

In our first main theorem, we construct near-optimal extractors for small-space sources in the polynomial error regime.

\begin{THM}\label{thm:MAIN:small-space-polylog-entropy}
	There exists a universal constant \(C>0\) such that for all \(n,k,s\in\N\) satisfying \(k\geq s\cdot\log^C(n)\), there exists an explicit extractor \(\Ext:\zo^n\to\zo^m\) for space \(s\) sources with min-entropy \(k\), which has output length \(m=(k/s)^{\Omega(1)}\) and error \(\epsilon=n^{-\Omega(1)}\).
\end{THM}

Thus, our extractor requires min-entropy \(k\geq s\cdot\log^C(n)\), which is an exponential improvement over the previous best requirement \cite{chattopadhyay2016extractors} of \(k\geq s^{1.1}\cdot2^{\log^{0.51}(n)}\). In particular, in the natural setting of sources sampled by \(s=\polylog(n)\) space algorithms, our extractor is the first construction that works for polylogarithmic entropy. Non-constructively, it is known that small-space extractors exist for min-entropy \(k\geq O(s + \log n + \log(1/\epsilon))\), and thus our result is nearly optimal when the desired error is at most polynomially small.

The key ingredient we use to prove \cref{thm:MAIN:small-space-polylog-entropy} is a new structural result, which establishes a connection between small-space sources and \emph{affine sources}. An affine source \(\X\) over \(n\) bits with min-entropy \(k\) is a distribution that is uniform over some (unknown) affine subspace of \(\F_2^n\). A long line of work has considered the problem of constructing extractors for affine sources \cite{gabizon2008deterministic,devos2010simple,bourgain2007construction,Yehudayoff11,Li11,rao2009low,li2016improved,CGL21}, and in this work we show that such extractors can also extract from small-space sources. In particular, we prove the following.

\begin{THM}\label{thm:MAIN:structural-result}
Let \(\X\) be a space \(s\) source over \(\zo^n\) with min-entropy \(k\). Then \(\X\) is \(2^{-\Omega(k)}\)-close to a convex combination of affine sources with min-entropy \(\Omega(\frac{k}{s\log(n/k)})\).
\end{THM}

By combining this structural result with the explicit affine extractor of Li \cite{li2016improved}, which works for \(\polylog(n)\) min-entropy and has polynomially small error, we immediately obtain \cref{thm:MAIN:small-space-polylog-entropy}. Furthermore, if we are only interested in outputting one bit with constant error, we can use the recent affine extractor of Chattopadhyay, Goodman, and Liao \cite{CGL21} to extract from small-space sources with min-entropy \(k\geq s\cdot\log^{2+o(1)}(n)\).


\subsubsection{Small-space extractors with exponentially small error}

While polynomially small error suffices for many applications, it is sometimes important to achieve negligible error in applications such as cryptography \cite{dodis2004possibility}. However, since the best low-error affine extractors require entropy \(k\geq\Omega(n/\sqrt{\log\log n})\) \cite{bourgain2007construction,Yehudayoff11,Li11}, \cref{thm:MAIN:structural-result} does not yield any new result in the negligible error setting.

In our next main result, we develop some new machinery in order to obtain improved low-error extractors for small-space sources. Until recently, the best extractors for such sources \cite{kamp2006small} required entropy \(k\geq Cn^{1-\gamma}s^{\gamma}\), where \(\gamma>0\) is some tiny constant and \(C\) is a large one. In \cite{adversarial-sources}, the entropy requirement was improved to \(k\geq Cn^{2/3+\delta}s^{1/3-\delta}\). We reduce this entropy requirement further, and prove the following.

\begin{THM}\label{thm:main:space}
For any fixed \(\delta\in(0,1/2]\) there is a constant \(C>0\) such that for all \(n,k,s\in\N\) satisfying \(k\geq Cn^{1/2+\delta}s^{1/2-\delta}\), there exists an explicit extractor \(\Ext:\zo^n\to\zo^m\) for space \(s\) sources of min-entropy \(k\), with output length \(m=n^{\Omega(1)}\) and error \(\epsilon=2^{-n^{\Omega(1)}}\).
\end{THM}

Observe that the line of improvements described above (from \cite{kamp2006small} to \cite{adversarial-sources} to \cref{thm:main:space}) is strict, since we always have \(s<n\) (or else the bounds are trivial). In particular, note that for, say \(s=n^\delta\) space, the entropy requirement has dropped from \(k\geq O(n^{1-\gamma})\) to \(k\geq O(n^{2/3+\delta})\) to \(k\geq O(n^{1/2+\delta})\). 

To prove \cref{thm:main:space}, we start with the standard approach \cite{kamp2006small} of reducing small-space sources to the class of \emph{adversarial sources} \cite{adversarial-sources}. Informally, an adversarial source \(\X\) consists of many independent sources, where only a few of them are guaranteed to be ``good'' (i.e., contain some min-entropy). Formally, an \((N,K,n,k)\)-adversarial source \(\X\) consists of \(N\) independent sources \(\X_1,\dots,\X_N\), each over \(n\) bits, with the guarantee that at least \(K\) of them have min-entropy at least \(k\). Such sources have applications in generating a (cryptographic) common random string in the presence of adversaries, and in harvesting randomness from unreliable sources.

To prove \cref{thm:main:space}, we explicitly construct significantly improved extractors for adversarial sources:

\begin{THM}\label{thm:main:adversarial}
There is a universal constant \(C>0\) such that for any fixed \(\delta>0\) and all sufficiently large \(N,K,n,k\in\N\) satisfying \(k\geq\log^Cn\) and \(K\geq N^\delta\), there exists an explicit extractor \(\Ext:(\zo^n)^N\to\zo^m\) for \((N,K,n,k)\)-adversarial sources, with output length \(m=k^{\Omega(1)}\) and error \(\epsilon=2^{-k^{\Omega(1)}}\).	
\end{THM}

Previously, the best extractor for this setting \cite{adversarial-sources} required \(K\geq N^{0.5+o(1)}\) good sources, and our improvement to \(K\geq N^\delta\) is crucial in obtaining better extractors for small-space sources. An added bonus is that our extractor construction is arguably much simpler compared to \cite{adversarial-sources}. 

To prove \cref{thm:main:adversarial}, we develop a simple new framework for extracting from adversarial sources by combining (i) a general type of combinatorial design; and (ii) a specific kind of leakage-resilient extractor \cite{kms,focs2020merged}. While such leakage-resilient extractors were recently constructed explicitly in \cite{focs2020merged}, the only known construction of such designs is probabilistic \cite{sts}.

Thus, the key ingredient we use to prove \cref{thm:main:adversarial}, and subsequently \cref{thm:main:space}, is the first explicit construction of such designs. In more detail, an \((n,r,s)\)-design is an \(r\)-uniform hypergraph over \(n\) vertices with pairwise hyperedge intersections of size \(<s\). To instantiate our framework, we need explicit \((n,r,s)\)-designs with small independence number\footnote{Recall that an \emph{independent set} in a hypergraph is a subset of vertices that contain no hyperedge, and the \emph{independence number} of a hypergraph is the size of its largest independent set.} \(\alpha\). Previously, Chattopadhyay, Goodman, Goyal and Li \cite{adversarial-sources} constructed \((n,3,2)\)-designs with independence number \(\alpha\leq O(n^{0.923})\). To obtain our improved extractors in \cref{thm:main:adversarial,thm:main:space}, we need designs with much smaller independence number. Our final main theorem constructs exactly such designs.

\begin{THM}\label{thm:main:designs}
	For all constants \(r\geq s\in\N\) with \(r\) even, there exist explicit \((n,r,s)\)-designs \((G_n)_{n\in\N}\) with independence number
	\[
	\alpha(G_n)\leq O(n^{\frac{2(r-s)}{r}}).
	\]
\end{THM}

\cref{thm:main:designs} gives the first derandomization of a result by R\"{o}dl and \v{S}inajov\'{a} \cite{sts}, and our explicit designs are optimal up to a factor of \(2\) in the power. We show that it is easy to extend \cref{thm:main:designs} to also work for odd \(r\) (up to a small loss in parameters), and we also show that our construction remains explicit for most \emph{super-constant} \(r,s\): we refer the reader to \cref{sec:designs} for more detail.

Finally, we can combine our explicit designs with the leakage-resilient extractors from \cite{focs2020merged} to obtain our improved adversarial sources (\cref{thm:main:adversarial}), which immediately yields our improved extractors for small-space sources (\cref{thm:main:space}). It is known that the technique of reducing small-space sources to adversarial sources has a barrier at min-entropy \(\sqrt{n}\)  (see \cref{rem:root-n-barrier}). Thus, the result in \cref{thm:main:space} has almost the best parameters one can hope to achieve using this technique.

\section{Overview of Techniques}\label{sec:overview}

We use this section to sketch the explicit constructions of our small-space extractors. We start with our low-error small space extractors (\cref{thm:main:space}) and the ingredients that go into it (\cref{thm:main:adversarial,thm:main:designs}). Then, we sketch the construction of our small-space extractor for polylogarithmic entropy (\cref{thm:MAIN:small-space-polylog-entropy}) and its key ingredient (\cref{thm:MAIN:structural-result}).

\subsection{Small-space extractors with exponentially small error}\label{sec:over_low_ext}

To construct our low-error small-space extractors, the first step is to use a standard reduction \cite{kamp2006deterministic} (which we slightly optimize) from small-space sources to adversarial sources. This reduction starts with the observation of \cite{kamp2006small} that if we chop up the small space source \(\X\) into \(t\) consecutive (equal-sized) chunks, and \emph{condition on any fixing} of the vertices reached at the end of each chunk \emph{in the random walk that generates} \(\X\), then these \(t\) chunks become \(t\) independent sources. Furthermore, if \(\X\) originally had \(k\) bits of entropy, then it follows from the entropy chain rule that \(\X\) will still have roughly \(k-st\) bits of entropy. A Markov argument then shows that at least a few of the \(t\) sources will have relatively high entropy. In other words, \(\X\) now looks like an adversarial source, and we may now focus on constructing (low-error) extractors for adversarial sources.

\paragraph{Improved low-error extractors for adversarial sources}
To construct our low-error extractors for adversarial sources, we develop a new framework that combines a certain type of \emph{leakage-resilient extractor} (LRE) with the \((n,r,s)\)-designs discussed earlier. An LRE for \(r\) sources offers the guarantee that its output looks uniform \emph{even conditioned on} the output of many \emph{leakage} functions, each called on up to \(r-2\) of the same inputs fed to the original LRE. Furthermore, recall that an \((n,r,s)\)-design is an \(r\)-uniform hypergraph over \(n\) vertices with pairwise hyperedge intersections of size \(<s\).

Now, given an \((N,K,n,k)\)-adversarial source \(\X\), we extract from it as follows, using an LRE and an \((N,r,r-1)\)-design \(G\) with independence number \(\alpha(G)<K\). First, we identify the vertices of our design with the \(N\) independent sources in \(\X\). Then, for each hyperedge in our design, we call a leakage-resilient extractor on the \(r\) sources it contains, and finish by taking the bitwise XOR over the outputs of the LRE calls. 

This construction successfully outputs uniform bits for the following reasons. Because \(\alpha(G)<K\), we are guaranteed that \emph{some} LRE call is given \emph{only} good sources. By the \emph{extractor} property of the LRE, this call will output uniform bits. Meanwhile, the \emph{bounded intersection} property of the \((N,r,r-1)\)-design, paired with the \emph{leakage-resilience} property of the LRE, guarantees that these uniform bits still look uniform \emph{even after taking their bitwise XOR with the outputs of all other LRE calls}. Using these ideas, we actually provide a slightly more general framework to combine \((N,r,s)\)-designs with LREs of various strength. Our framework leverages the ``\emph{activation vs. fragile correlation}'' paradigm introduced in \cite{adversarial-sources}, yet it is able to do so in a much more simple, general, and effective way, by combining two very general pseudorandom objects: LREs and designs.

To make our framework explicit, we will need explicit LREs and explicit designs with small independence number. Our explicit LREs will come from the work of Chattopadhyay et al. \cite{focs2020merged}, where they gave the first explicit LREs that work for entropy \(k=o(n)\), and in fact their LREs work for entropy \(k\geq\polylog(n)\). Thus all that remains is to provide an explicit construction of designs with small independence number.\footnote{Explicit \((N,3,2)\)-designs with independence number \(\alpha<O(N^{0.923})\) were constructed in \cite{adversarial-sources}. However, we need more general \((N,r,r-1)\)-designs to push the independence number low enough to obtain our desired adversarial extractors, and (to the best of our knowledge) no such explicit designs were known prior to our work.} We provide such a construction in this paper, and sketch it below.

\paragraph{Explicit designs with small independence number} In order to construct our \((n,r,s)\)-designs \(G=(V,E)\), we start with a linear code \(Q\subseteq\F_2^n\) of distance \(d>2(r-s)\), and then restrict it to the set \(Q_r\subseteq Q\) of elements in \(Q\) that have Hamming weight \(r\). Our design \(G=(V,E)\) is constructed by identifying \(V\) with \([n]\), and by creating a hyperedge for each \(x\in Q_r\) in the natural way. The distance of the code and the definition of \(Q_r\) immediately guarantees that \(G\) is an \((n,r,s)\)-design.

In order to upper bound the independence number \(\alpha(G)\) of our design, we observe that any independent set in \(G\) corresponds to a subcube \(S\subseteq\F_2^n\) that contains no vector in \(Q\) of weight \(r\); in other words, since \(Q\) is a \emph{linear} code, this means that the \emph{subspace} \(T^\ast:=S\cap Q\) has no vector of Hamming weight \(r\). If our linear code \(Q\) had very high dimension, then even if the subcube \(S\) was relatively small, we would have found a relatively large subspace \(T^\ast\) containing no vector of Hamming weight \(r\). But intuitively, it seems like as the dimension of a subspace grows large enough, at some point it must be guaranteed to have such a vector. It turns out this is true, and it follows immediately from Sidorenko's recent bounds \cite{sidorenko2018extremal,sidorenko2020generalized} on the size of sets in \(\F_2^n\) containing no \(r\) elements that sum to zero. Thus if \(Q\) has large enough dimension, \(S\) cannot be too large, and thus neither can \(\alpha(G)\). All that remains is to explicitly construct (the weight-\(r\) vectors of) a high-dimensional linear code \(Q\subseteq\F_2^n\) with distance \(d>2(r-s)\), which can easily be done using BCH codes \cite{bch-codes-bc,bch-codes-h}.

\subsection{Small-space extractors for polylogarithmic entropy}

Unfortunately, it is impossible to extract from small-space sources with entropy \(k<\sqrt{n}\) using a reduction of the previous type (i.e., to adversarial sources), since setting \(t\geq\sqrt{n}\) will leave \(k-st\leq k-1\cdot\sqrt{n}<0\) bits of entropy after the above fixing, while setting \(t<\sqrt{n}\) will produce a chunk of size \(n/t>\sqrt{n}>k\), which could hold all of the entropy and thus make extraction impossible. To circumvent this barrier, we provide a new reduction from small-space sources to \emph{affine sources}. This reduction bypasses the \(\sqrt{n}\) barrier by \emph{adaptively} choosing vertices to fix: this was not possible above, because such adaptive fixings can produce independent sources of unknown and varying lengths, which cannot be captured by adversarial sources. We describe our new reduction in more detail below.

\paragraph{A reduction from small-space sources to affine sources} Our new reduction from small-space sources to affine sources starts the same way as before: by fixing \(t\) vertices in the random walk generating the space \(s\) source \(\X\), to create \(t\) independent sources with roughly \(k-st\) bits of total entropy. The key idea now is to use a nice observation of \cite{adversarial-sources}, which says that \emph{any} source with entropy at least \(1\) is a convex combination of \emph{affine sources} with entropy \(1\). Given this observation, we can say that as long as \(t^\pr\) of the \(t\) independent sources have \emph{just one bit of entropy}, then \(\X\) currently looks like a convex combination of affine sources with min-entropy \(t^\pr\).

On the other hand, if \emph{no} \(t^\pr\) of the \(t\) independent sources have just one bit of entropy, then the \(k-st\) remaining bits of entropy must be \emph{very} highly concentrated on the \(t^\pr-1\) most entropic independent sources. In this case, we can simply recursively apply the reduction on these \(t^\pr-1\) independent sources. Because the entropy rate increases on each recursive call, we know the recursion must eventually stop, or else we will end up with a source with entropy rate exceeding \(1\), a contradiction. Thus, via a \emph{win-win argument}, we are able to show that \(\X\) is a convex combination of affine sources with entropy \(t^\pr\).

We show that even if \(\X\) starts with entropy just \(k\geq\polylog(n)\), our resulting affine source will have almost all of the entropy of the original source; namely, \(t^\pr\) will barely be smaller than \(k\). We are able to achieve such an efficient reduction for two reasons. First, our use of \emph{affine sources} allows an \emph{adaptive} and \emph{recursive} reduction that bypasses the \(k\geq\sqrt{n}\) entropy barrier arising from existing reductions to source types of fixed lengths (like \emph{total-entropy sources} \cite{kamp2006small} and \emph{adversarial sources} \cite{adversarial-sources}). Second, our reduction to a sequence of \emph{\(t^\pr\) independent sources with entropy \(1\)} (which we argue is an affine source with entropy \(t^\pr\) using the observation of \cite{adversarial-sources}) results in a \emph{negligible} amount of lost entropy from each recursive step, whereas similar recursive reductions  to \emph{a constant number of sources with relatively high entropy} \cite{chattopadhyay2016extractors} are forced to lose much more entropy in each such step. As a result, we are able to bypass the \(k\geq2^{\sqrt{\log n}}\) entropy barrier of \cite{chattopadhyay2016extractors}.
 
 Finally, we note that by carefully tracking the random variables that pop up in our recursion, we are able to describe all of the fixings that occur throughout the recursion \emph{by the fixing of a single random variable}. As a result, we only need to apply the chain rule for min-entropy (\cref{lem:entropy-drop}) \emph{once}, which keeps the error of our reduction very low: \(2^{-\Omega(k)}\), compared to an error of \(2^{-k^{\Omega(1)}}\) in the recursive reduction of \cite{chattopadhyay2016extractors}.

\paragraph{Organization} In \cref{sec:prelims} we provide several preliminaries.  In the remainder of our paper, we follow a \emph{bottom-up} strategy for presenting our main results. In \cref{sec:designs}, we provide an explicit construction of designs with small independence number, proving \cref{thm:main:designs}. In \cref{sec:adversarial-sources}, we show how to combine these designs with leakage-resilient extractors to create a new, simple framework for extraction from adversarial sources. By instantiating our framework with our explicit designs and the explicit leakage-resilient extractors of \cite{focs2020merged}, we obtain our improved extractors for adversarial sources, \cref{thm:main:adversarial}. In \cref{sec:small-space}, we provide the standard reduction from small-space sources to adversarial sources for completeness, and we apply our adversarial extractors (\cref{thm:main:adversarial}) to obtain our small-space extractors with exponentially small error, \cref{thm:main:space}. In \cref{sec:space:polylog-entropy}, we provide our \emph{new reduction} from small-space sources to \emph{affine sources} (\cref{thm:MAIN:structural-result}) and apply the affine extractor of Li \cite{li2016improved} to obtain our small-space extractors for polylogarithmic entropy, \cref{thm:MAIN:small-space-polylog-entropy}. We conclude with some remarks and present some open problems in \cref{sec:conclusions}.

\section{Preliminaries}\label{sec:prelims}

\paragraph{General notation} Given two strings \(x,y\in\zo^m\), we let \(x\oplus y\) denote their bitwise XOR. For a number \(n\in\N\), \([n]\) denotes the interval \([1,n]\subseteq\N\). We let \(\circ\) denote string concatenation, and for a collection \(\{x_i:i\in I\}\) indexed by some finite set \(I\), we let \((x_i)_{i\in I}\) denote the concatenation of all strings \(x_i,i\in I\). If \(I\) is already equipped with some total order, this is used to determine the concatenation order; otherwise, \(I\) is arbitrarily identified with \([|I|]\) to induce a total ordering. Given a domain \(\mathcal{D}\), and some string \(x\in\mathcal{D}^N\), we let \(x_i\in\mathcal{D}\) denote the value at the \(i^\text{th}\) coordinate of \(x\). Given a subset \(S\subseteq[N]\), we let \(x_S := (x_i)_{i\in S}\). Even if \(\mathcal{D}=\mathcal{R}^n\) for some other domain \(\mathcal{R}\) and number \(n\in\N\), the definition of \(x_S\in\mathcal{D}^{|S|}\) does not change.

\paragraph{Basic coding theory and extractor definitions} We let \(\F_2\) denote the finite field of size two, and we let \(\F_2^n\) denote a vector space over this field. The \emph{Hamming weight} of a vector \(x\in\F_2^n\) is defined as \(\Delta(x):=\#\{i\in[n] : x_i=1\}\), and the \emph{Hamming distance} between two vectors \(x,y\in\F_2^n\) is defined as \(\Delta(x,y):=\Delta(x-y)\), where the subtraction is over \(\F_2\). The \emph{standard basis vectors} in \(\F_2^n\) is the collection \(\mathcal{E}^\ast:=\{e_i\}_{i\in[n]}\), where \(e_i\in\F_2^n\) holds a 1 at coordinate \(i\) and \(0\) everywhere else, and a \emph{subcube} is a subspace spanned by some subset of \(\mathcal{E}^\ast\). An \((n,k,d)\)-code is a subset \(Q\subseteq\F_2^n\) of size \(2^k\) with the guarantee that any two distinct points \(x,y\in Q\) have Hamming distance \(\Delta(x,y)\geq d\). A linear \([n,k,d]\)-code is simply an \((n,k,d)\) code that is a subspace. Finally, we say that a source \(\X\) over \(\zo^n\) is an \((n,k)\) source if it has min-entropy at least \(k\), and we say that an extractor \(\Ext\) an \(N\)-source extractor for entropy \(k\) if it is an extractor for a family of sources \(\mathcal{X}\), where each \(\X\in\mathcal{X}\) consists of \(N\) independent \((n,k)\) sources.

\paragraph{Discrete probability} In general, for a random variable \(\X:\Omega\to V\), we are only concerned with the distribution over \(V\) induced by \(\X\). We will therefore typically not define the outcome space \(\Omega\), and can assume it has any form we like (so long as the distribution induced by \(\X\) does not change). Given random variables \(\X,\Y\) and any \(y\in\supp(\Y)\), we let \((\X\mid\Y=y)\) denote a random variable that takes value \(x\) with probability \(\Pr[\X=x\mid\Y=y]\). Given a random variable \(\X\) and a family of random variables \(\mathcal{Y}\), we say that \(\X\) is a \emph{convex combination} of random variables from \(\mathcal{Y}\) if there exists a random variable \(\mathbf{Z}\) such that for each \(z\in\supp(\mathbf{Z})\), it holds that \((\X\mid\mathbf{Z}=z)\in\mathcal{Y}\). We define the \emph{statistical distance} between two random variables \(\X,\Y\) over \(V\) as
\[
|\X-\Y|:=\max_{S\subseteq V}|\Pr[\X\in V]-\Pr[\Y\in V]|=\frac{1}{2}\sum_{v\in V}|\Pr[\X=v]-\Pr[\Y=v]|,
\]
and we say that \(\X,\Y\) are \emph{\(\epsilon\)-close} if \(|\X-\Y|\leq\epsilon\). Given these definitions, the following two standard facts are easy to show, and are extremely useful.

\begin{fact}\label{fact:stat-dist:add-constant}
For any random variable \(\X\sim\zo^m\) and any constant \(c\in\zo^m\), it holds that
\[
|\X-\U_m|=|(\X\oplus c)-\U_m|.
\]
\end{fact}

\begin{fact}\label{fact:stat-dist:convex-combination}
For any random variables \(\X,\Y\), where \(\X\sim\zo^m\), it holds that
\[
|\X-\U_m|\leq\E_{y\sim\Y}[|(\X\mid\Y=y)-\U_m|].
\]
\end{fact}

Finally, we will need the following standard lemma about conditional min-entropy.
\begin{lemma}[\hspace{1sp}\cite{mw97}]\label{lem:entropy-drop}
Let \(\X,\Y\) be random variables such that \(\Y\) can take at most \(\ell\) values. Then for any \(\epsilon>0\), it holds that
\[
\Pr_{y\sim\Y}[H_\infty(\X\mid\Y=y)\geq H_\infty(\X)-\log\ell-\log(1/\epsilon)]\geq1-\epsilon.
\]
\end{lemma}

\section{Explicit extremal designs via slicing codes and zero-sum sets}\label{sec:designs}

In this section, we will construct our explicit designs and thereby prove \cref{thm:main:designs}. Before we state the formal theorem and proof, we begin with some background and discussion on \((n,r,s)\)-designs.

\subsection{Background and discussion}

A \emph{combinatorial design} is a special type of \emph{well-balanced set system}, where each set has the same size, and no two sets intersect at too many points. More formally, we say that an \(r\)-uniform hypergraph \(G=(V,E)\) over \(n\) vertices is an \emph{\((n,r,s)\)-design}, or \emph{\((n,r,s)\)-partial Steiner system}, if \(|e_1\cap e_2|<s\) for all distinct \(e_1,e_2\in E\). Beyond the fact that they are pseudorandom objects themselves, it turns out that \((n,r,s)\)-designs enjoy several interesting applications in pseudorandomness.


A notable application of designs is in the seminal work of Nisan and Wigderson \cite{nisan1994hardness}, where they are used to construct \emph{pseudorandom generators} (PRGs). In this application, the authors require (and provide) explicit designs that are \emph{extremal} in the sense that they have a large number of hyperedges. More recently, explicit designs of a different extremal flavor have been used in the construction of extractors: in \cite{adversarial-sources}, Chattopadhyay, Goodman, Goyal, and Li show how to construct extractors for adversarial sources using explicit partial Steiner triple systems (\((n,3,2)\)-designs) with \emph{small independence number}.

Given these applications, it is natural to ask about the smallest possible independence number of more general \((n,r,s)\)-designs. R\"{o}dl and \v{S}inajov\'{a} answered this question in 1994, proving the following:

\begin{theorem}[\hspace{1sp}\cite{sts}]\label{thm:rodl-sinajova}
Given any \(n\geq r\geq s\in\N\) with \(r\geq2\), there exists an \((n,r,s)\)-design \(G\) with independence number
\[
\alpha(G)\leq C_{r,s}\cdot n^{\frac{r-s}{r-1}}(\log n)^{\frac{1}{r-1}},
\]
where \(C_{r,s}=C(r,s)\) depends only on \(r,s\).
\end{theorem}

In fact, they also showed this result is tight up to the term \(C_{r,s}\) that depends only on \(r,s\).

In order to prove \cref{thm:rodl-sinajova}, R\"{o}dl and \v{S}inajov\'{a} apply the Lov\'{a}sz Local Lemma to show that a \emph{random} \(r\)-uniform hypergraph is such a design. Thus, while their result proves the existence of such designs, it does not provide an explicit way to construct them  - and, unfortunately, an explicit construction is needed if one hopes to apply this result to construct other explicit objects (like extractors). Furthermore, all subsequent work appears to focus on improving the term \(C_{r,s}\) \cite{eustis2013independence,eustis2013hypergraph} or extending their result to more general types of designs \cite{grable1995minimum,kostochka2014independent,tian2018bounding}, while still relying on probabilistic constructions.

In this section, we will provide explicit constructions of these extremal designs. Our designs give the first derandomization of \cref{thm:rodl-sinajova}, and differ from the optimal bound by just a square.

\subsection{Main theorem about explicit designs}
We are now ready to state our main theorem that describes our construction of explicit designs with small independence number.

\begin{theorem}[\cref{thm:main:designs}, formal version]\label{thm:formal:main-design}
There exists an Algorithm \(\mathcal{A}\) such that given any \(n\geq r\geq s\in\N\) as input with \(r\) an even number, \(\mathcal{A}\) runs in time \(\poly\left(\binom{n}{r}\right)\) and outputs an \((n,r,s)\)-design \(G\) with independence number
\begin{align}\label[ineq]{ineq:main-design-bound}
\alpha(G)\leq C_{r,s}\cdot n^{\frac{2(r-s)}{r}},
\end{align}
where \(C_{r,s}=C\cdot r^4\) for some universal constant \(C\geq1\).
\end{theorem}
\begin{remark}\label{odd-designs}
It is easy to extend \cref{thm:formal:main-design} to construct \((n,r,s)\)-designs \((G_n)_{n\in\N}\) with odd \(r\), at the expense of a small loss in parameters: simply construct an \((n,r+1,s)\)-design \(G_n^\pr\) using \cref{thm:formal:main-design}, and remove an arbitrary vertex from each hyperedge to create \(G_n\). \(G_n\) will be an \((n,r,s)\)-design, and will have independence number \(\alpha(G_n)\leq C_{r+1,s}\cdot n^{\frac{2(r+1-s)}{r+1}}\).
\end{remark}

For all constants \(r\geq s\in\N\) with \(r\) even, \cref{thm:formal:main-design} constructs an explicit family of \((n,r,s)\)-designs \((G_n)_{n\in\N}\) with small independence number. Like the non-explicit designs of \cref{thm:rodl-sinajova} from \cite{sts}, our derandomization focuses on the case where \(r,s\) are constant. However, it turns out that even for most \emph{super-constant} \(r,s\), our algorithm is still efficient. In particular, before proving \cref{thm:formal:main-design}, we make (and quickly prove) the following remark.

\begin{remark}\label{rem:super-constant}
Let \(\mathcal{A}\) be the algorithm from \cref{thm:formal:main-design}, and let \(m=m(n,r,s)\) be the number of hyperedges in the design produced by \(\mathcal{A}\) on input \((n,r,s)\). Then for any functions \(r=r(n),s=s(n)\), Algorithm \(\mathcal{A}\) is guaranteed to run in time \(\poly(n,m)\) over the collection \(\mathcal{I}=\{(n,r(n),s(n))\}_{n\in\N}\) as long as at least one of the following holds:
\begin{itemize}
\item The functions \(r,s\) are constant: \(r(n)=O(1)\) and \(s(n)=O(1)\); or
\item There is a constant \(\epsilon>0\) such that \cref{ineq:main-design-bound} is bounded above by \(O(n^{1-\epsilon}),\forall (n,r,s)\in\mathcal{I}\).
\end{itemize}
\end{remark}

The first bullet in \cref{rem:super-constant} reiterates the fact that the algorithm in \cref{thm:formal:main-design} is efficient when \(r,s\) are constant. The second bullet gives a more general remark on the performance of Algorithm \(\mathcal{A}\) on super-constant \(r,s\): it says that as long as \cref{thm:formal:main-design} gave a ``non-trivial'' bound on the independence number in the first place, then the algorithm will run efficiently. This effectively covers all ``interesting'' regimes of \(r,s\): indeed, the main application of selecting non-constant \(r,s\) would be to achieve independence bounds that are stronger than those achieved by constant \(r,s\) (and any constant \(r,s\) that achieve \(\alpha(G)<n\) in \cref{thm:formal:main-design} in fact achieve the second bullet).

To prove that Algorithm \(\mathcal{A}\) is efficient given the condition in the second bullet, we use standard bounds on Tur\'{a}n numbers. The Tur\'{a}n number \(T(n,\beta,r)\) is defined as the fewest number of edges in an \(r\)-uniform hypergraph with no independent set of size \(\beta\), and it is known \cite{sidorenko1995we} that \(T(n,\beta,r)\geq{n\choose r}/{\beta\choose r}\). Thus, the second bullet implies the number of edges, \(m=m(n,r,s)\), in the design is at least
\[T(n,Cn^{1-\epsilon}+1,r)\geq T(n,n^{1-\epsilon/2},r)\geq{n\choose r}/{n^{1-\epsilon/2}\choose r}\geq{n\choose r}/{n\choose r}^{1-\epsilon/4}\geq{n\choose r}^{\epsilon/4},\]
where we use the observation that the Tur\'{a}n number is non-increasing in its second argument, the fact that we can assume \(n,r\) are sufficiently large (since otherwise the efficiency claim is trivial), and a simple application of Stirling's formula. Thus, Algorithm \(\mathcal{A}\) runs in time \({n\choose r}=\poly(n,m)\). In fact, since we gave a lower bound on \(m\) based on the independence number, it trivially holds that \emph{any} algorithm that achieves independence numbers as small as \(\mathcal{A}\) must output \(m\) edges, meaning that the runtime of \(\mathcal{A}\) is optimal up to constant powers. This completes our discussion on \cref{rem:super-constant}.

\subsection{Proof of \cref{thm:formal:main-design}}\label{sec5:proof}
We now turn to proving \cref{thm:formal:main-design}. We start with the simple observation that hypergraphs over \(n\) vertices can be identified with subsets of \(\F_2^n\). In particular, any subset \(T\subseteq\F_2^n\) induces a hypergraph \(G_T=(V,E)\) in the following way: identify \(V\) with \([n]\), and for each \(x\in T\) add a hyperedge \(e\subseteq[n]\) to \(E\) that contains exactly the coordinates that take the value \(1\) in \(x\). Using this correspondence, we can instead focus on constructing special subsets of \(\F_2^n\), and thereby leverage the tools of linear algebra and coding theory.

To obtain our designs, we will need to explicitly construct a subset \(T\subseteq\F_2^n\) such that (1) \(G_T\) is an \((n,r,s)\)-design; and (2) \(G_T\) has small independence number. We can make sure this happens via the following two simple facts, which describe how these hypergraph properties can be identified with properties of subsets in \(\F_2^n\).

\begin{fact}\label{fact:constant-weight-code}
For any subset \(T\subseteq\F_2^n\), the hypergraph \(G_T\) is an \((n,r,s)\)-design if and only if (i) every \(x\in T\) has \(\Delta(x)=r\); and (ii) any two distinct \(x,y\in T\) have \(\Delta(x,y)>2(r-s)\).
\end{fact}
\begin{proof}
The two conditions are sufficient because the first one guarantees that \(G_T\) will be \(r\)-uniform, and the second one guarantees that any two edges in \(G_T\) intersect at \(<s\) points. They are both necessary because if the first does not hold, \(G_T\) will not be \(r\)-uniform, and if the first holds but the second does not, then two edges will end up sharing \(\geq s\) points.
\end{proof}

\begin{fact}\label{fact:subcube}
For any subset \(T\subseteq \F_2^n\), the hypergraph \(G_T\) has independence number \(\alpha(G_T)<\ell\) if and only if every subcube \(A\subseteq\F_2^n\) of dimension at least \(\ell\) has at least one point in \(T\).
\end{fact}
\begin{proof}
If \(\alpha(G_T)\geq\ell\), there is an independent set \(S\subseteq V=[n]\) of size at least \(\ell\), and thus the subcube \(A:=span(\{e_i\}_{i\in S})\) of dimension \(\ell\) has no points in \(T\). If there is a subcube \(A\subseteq\F_2^n\) of dimension \(\ell\) with no points in \(T\), the set \(S\subseteq[n]\) indexing the standard basis vectors that span \(A\) must have size \(\ell\) and constitute an independent set in \(G_T\).
\end{proof}

By \cref{fact:constant-weight-code} and \cref{fact:subcube}, we see that the task of constructing an \((n,r,s)\)-design \(G\) with small independence number is \emph{equivalent} to the task of constructing a subset \(T\subseteq\F_2^n\) with the following \emph{three properties}:
\begin{enumerate}
\item \(T\) lies in the Hamming slice \(\Delta_r:=\{x\in\F_2^n : \Delta(x)=r\}\),
\item Points in \(T\) have pairwise Hamming distance \(>2(r-s)\), and
\item Any subcube of \emph{relatively small} dimension intersects \(T\).
\end{enumerate}

In order to construct a set \(T\subseteq\F_2^n\) with these three properties, we use connections to \emph{coding theory} and \emph{zero-sum problems}. In particular, recall that an \((n,k,d)\)-code is a subset \(Q\subseteq\F_2^n\) of size \(2^k\) with the guarantee that any two distinct points \(x,y\in Q\) have Hamming distance \(\Delta(x,y)\geq d\). Thus, if we take any \((n,k,d)\)-code \(Q\subseteq\F_2^n\) with \(d>2(r-s)\) and intersect it with the Hamming slice \(\Delta_r\), we obtain a set \(T=Q\cap\Delta_r\) that enjoys properties (1) and (2). In order to endow it with property (3), we will need to start with some code \(Q\) such that for any relatively large subcube \(S\), the set \(S\cap T = S\cap (Q\cap\Delta_r)=(S\cap Q)\cap\Delta_r\) is non-empty.

The trick here is to start with a \emph{linear} code \(Q\). A \emph{linear} \([n,k,d]\)-code \(Q\subseteq\F_2^n\) is simply an \((n,k,d)\) code that is also a subspace. The condition \((S\cap Q)\cap\Delta_r\neq\emptyset\) required for property (3) now becomes more concrete: since \(Q\) is a subspace, \(S\cap Q\) is also a subspace, and thus we can make sure it contains some vector of Hamming weight \(r\) as long as we can show that \emph{every} large subspace contains such a vector. In particular, defining \(\Lambda_r(n)\) to be the dimension of the largest subspace \(R\subseteq\F_2^n\) containing no vector of Hamming weight \(r\), we prove the following lemma.

\begin{lemma}\label{lem:conditional-main-design}
If \(Q\subseteq\F_2^n\) is a linear \([n,k,d]\)-code with \(d>2(r-s)\), then the hypergraph \(G_{Q\cap\Delta_r}\) is an \((n,r,s)\)-design with independence number \(\alpha=\alpha(G_{Q\cap\Delta_r})\) that obeys the following inequality:
\[
\alpha-\Lambda_r(\alpha)\leq n-k
\]
\end{lemma}
\begin{proof}
It follows immediately from \cref{fact:constant-weight-code} that \(G_{Q\cap\Delta_r}\) is an \((n,r,s)\)-design. By \cref{fact:subcube}, there is a subcube \(A=span(e_{i_1},\dots,e_{i_\alpha})\subseteq\F_2^n\) of dimension \(\alpha\) that does not intersect \(Q\cap\Delta_r\). Thus, if we define \(A^\pr:=A\cap Q\), then \(A^\pr\) contains no vector of Hamming weight \(r\), and furthermore it has dimension \(dim(A^\pr)=dim(A\cap Q)\geq dim(A)+dim(Q)-n=\alpha+k-n\). Notice now that if we define the projection \(\pi:\F_2^n\to\F_2^\alpha\) as the map \((x_1,\dots,x_n)\mapsto(x_{i_1},\dots,x_{i_\alpha})\), then the subset \(\pi(A^\pr)\) is still a subspace (albeit now of \(\F_2^\alpha\)) of dimension \(dim(\pi(A^\pr))\geq\alpha+k-n\) containing no vector of Hamming weight \(r\). Thus, by definition of \(\Lambda_r\), it must hold that \(\alpha+k-n\leq dim(\pi(A^\pr))\leq\Lambda_r(\alpha)\).
\end{proof}

In order to construct an \((n,r,s)\)-design from \cref{lem:conditional-main-design} with the smallest possible independence number \(\alpha\), we will want an explicit \([n,k,d>2(r-s)]\)-linear code with the largest possible dimension \(k\), along with a strong upper bound on \(\Lambda_r(n)\). We start with the latter.

Getting a good upper bound on \(\Lambda_r(n)\) is closely related to the theory of \emph{zero-sum problems}. In this field, one parameter of great interest is the (generalized) \emph{Erd\H{o}s-Ginzburg-Ziv constant}(s) of a finite abelian group. Given \(n\geq r\in\N\) where \(r\) is even, this parameter is defined for \(\F_2^n\) as the smallest integer \(s_r(n)\) such that any \emph{sequence} of \(s_r(n)\) values in \(\F_2^n\) contains a subsequence of length \(r\) that sums to zero. For our application, it will be more convenient to use an almost identical parameter \(\beta_r(n)\), defined as the size of the largest \emph{subset} of \(\F_2^n\) containing no \(r\) elements that sum to zero. Using slightly different terminology, the relationship between \(\beta_r(n)\) and \(\Lambda_r(n)\) was shown in \cite{sidorenko2020generalized}. We include it here, in our language, for completeness.

\begin{lemma}[\hspace{1sp}\cite{sidorenko2020generalized}]\label{lem:subspace-hamming}
For every \(n\geq r\in\N\) where \(r\) is even,
\[
\beta_r(n-\Lambda_r(n))\geq n.
\]
\end{lemma}
\begin{proof}
Let \(R\subseteq\F_2^n\) be a subspace of dimension \(k:=\Lambda_r(n)\) that contains no vector of Hamming weight \(r\), and define \(d:=n-k\). Let \(v_1,\dots,v_d\) be a basis for the orthogonal complement of \(R\), and define the matrix \(M\in\F_2^{d\times n}\) so that its \(i^\text{th}\) row is \(v_i\). Notice that \(R\) contains exactly the solutions to \(Mx=0\), and thus \(R\) has a vector of Hamming weight \(r\) if and only if there are \(r\) columns in \(M\) that sum to zero. By definition of \(R\), we know \(R\) has no such vector, and thus \(n\leq\beta_r(d)=\beta_r(n-\Lambda_r(n))\).
\end{proof}

To get a good upper bound on \(\Lambda_r(n)\), we need a good upper bound on \(\beta_r(n)\). In 2018, Sidorenko provided a very strong bound of this type:

\begin{theorem}[\hspace{1sp}\cite{sidorenko2018extremal}, Theorem 4.4]\label{thm:sidorenko-main}
There is a universal constant \(C>0\) such that for every \(n,r\in\N\) where \(r\) is even,
\[
\beta_{r}(n)\leq C\cdot r^3\cdot2^{2n/r}.
\]
\end{theorem}
By plugging this bound into \cref{lem:subspace-hamming}, we get the following corollary.
\begin{corollary}[\hspace{1sp}\cite{sidorenko2020generalized}]\label{cor:subspace-hammer}
There is a universal constant \(C>0\) such that for any \(n\geq r\in\N\) where \(r\) is even, the largest subspace \(S\subseteq\F_2^n\) with no vector of Hamming weight \(r\) has dimension
\[
\Lambda_r(n)\leq n-(r\log n -3r\log r - r\log C)/2.
\]
\end{corollary}

We are finally ready to prove our main design lemma, which reduces the problem of constructing \((n,r,s)\)-designs with small independence number to constructing high-dimensional linear codes.

\begin{lemma}[Main design lemma]\label{lem:main:design}
There is a universal constant \(C>0\) such that for every \(n\geq r\geq s\) with \(r\) even, if \(Q\subseteq\F_2^n\) is a linear \([n,k,d]\)-code with \(d>2(r-s)\), then \(G_{Q\cap\Delta_r}\) is an \((n,r,s)\)-design with independence number
\[
\alpha(G_{Q\cap\Delta_r})\leq C\cdot r^3\cdot2^{2(n-k)/r}.
\]
\end{lemma}
\begin{proof}
Simply plug the bound on \(\Lambda_r(\alpha)\) from \cref{cor:subspace-hammer} into \cref{lem:conditional-main-design}.
\end{proof}

To complete the proof of \cref{thm:formal:main-design}, we now just need to explicitly construct a linear code with very high dimension. In 1959-1960, Bose, Ray-Chaudhuri \cite{bch-codes-bc}, and Hocquenghem \cite{bch-codes-h} explicitly constructed codes of exactly this type (see \cite{guruswami2010notes} for a great exposition of these codes, which are known as \emph{BCH codes}). In particular, they proved the following theorem.

\begin{theorem}[\hspace{1sp}\cite{bch-codes-bc,bch-codes-h}]\label{thm:main:bch-codes}
For every \(m,t\in\N\), there exists an \([n,k,d]\)-linear code \(\mathbf{BCH}_{m,t}\subseteq\F_2^n\) with block length \(n=2^m-1\), dimension \(k\geq n-mt\), and distance \(d>2t\). Furthermore, there exists an Algorithm \(\mathcal{B}\) that given any \(m,t\in\N\) and \(x\in\F_2^{n}\) as input, checks if \(x\in\mathbf{BCH}_{m,t}\) in \(\poly(n)\) time.
\end{theorem}

By instantiating \cref{lem:main:design} with \cref{thm:main:bch-codes}, we can finally prove \cref{thm:formal:main-design}.

\begin{proof}[Proof of \cref{thm:formal:main-design}]
We start by assuming that \(n=2^m-1\) for some \(m\in\N\). Then, we let \(t=r-s\), and use \cref{thm:main:bch-codes} to define the \([n,k,d]\)-linear code \(Q:=\mathbf{BCH}_{m,t}\subseteq\F_2^n\), where \(k\geq n-mt=n-m(r-s)\) and \(d>2t=2(r-s)\). Algorithm \(\mathcal{A}\) will simply output the hypergraph \(G_{Q\cap\Delta_r}\). By \cref{lem:main:design}, we know that \(G_{Q\cap\Delta_r}\) is an \((n,r,s)\)-design with independence number
\[
\alpha(G_{Q\cap\Delta_r})\leq C\cdot r^3\cdot 2^{2(n-k)/r}\leq C\cdot r^3\cdot 2^{2mt/r}=C\cdot r^3\cdot (2^m)^{2(r-s)/r}\leq 2C\cdot r^3\cdot n^{2(r-s)/r}.
\]
Furthermore, note that \(G_{Q\cap\Delta_r}\) can be constructed in \(\poly({n\choose r})\) time if \(Q\cap\Delta_r\) can be constructed in \(\poly({n\choose r})\) time, and this can be done by simply checking (and appropriately including) whether each of the \({n\choose r}\) elements in \(\Delta_r\) belong to \(Q\), using Algorithm \(\mathcal{B}\) from \cref{thm:main:bch-codes}.

If \(n\) is of the form \(2^m\), we can follow the previous procedure to draw hyperedges around the first \(n-1\) vertices, and then add one more isolated vertex (contained in no edges) at the end to finish the hypergraph. Clearly we will still have \(\alpha(G_{Q\cap\Delta_r})\leq3C\cdot r^3\cdot n^{2(r-s)/r}\).

If \(n\) is not of the form \(2^m-1\) nor \(2^m\), then it can be written as a sum \(x_02^0 + \dots + x_d2^d\) over \(\N\), where \(d=\lceil\log n\rceil\) and each \(x_i\in\zo\). We can then follow the most recent procedure to construct a graph \(G_i\) over \(2^i\) vertices separately for each nonzero \(x_i\). The final graph \(G=\bigcup_i G_i\) is clearly still an \((n,r,s)\)-design, and it has independence number
\begin{align*}
\alpha(G)=\sum_{i:x_i=1}\alpha(G_i)\leq\sum_{0\leq i\leq\lceil\log n\rceil}3C\cdot r^3\cdot (2^i)^{2(r-s)/r}
&=3C\cdot r^3\sum_{0\leq i\leq\lceil\log n\rceil}(2^i)^{2(r-s)/r}\\
&=3C\cdot r^3\cdot\frac{(2^{\lceil\log n\rceil+1})^{2(r-s)/r}-1}{2^{2(r-s)/r}-1}.
\end{align*}
It is straightforward to verify that for a large enough universal constant \(C^\pr\), the above fraction is bounded above by \(C^\pr\cdot r\cdot n^{2(r-s)/r}\), which completes the proof.
\end{proof}

\section{Extractors for adversarial sources via designs and LREs}\label{sec:adversarial-sources}
Perhaps the most popular model of seedless extraction is to assume that each source \(\X \) actually consists of several \emph{independent sources} \(\X=(\X_1,\X_2,\dots,\X_N)\), each guaranteed to have some min-entropy. A long line of work has focused on constructing extractors for this setting \cite{chor1988unbiased,barak2006extracting,li2015three,cohen2016local,CZ19,Li19}, and has culminated in extractors with a near-optimal entropy requirement \cite{Li19}. Recently, the idea of generalizing this model to allow for \emph{bad sources} with \emph{no entropy guarantee} and/or \emph{limited dependence} has received considerable attention \cite{aggarwal2020extract,adversarial-sources,ball2019randomness}. Motivated by applications in generating a (cryptographic) common random string in the presence of adversaries, and in harvesting randomness from unreliable sources, Chattopadhyay, Goodman, Goyal, and Li \cite{adversarial-sources} introduced the class of \emph{adversarial sources}:

\begin{definition}[Adversarial sources]\label{def:intro:adversarial-sources}
A source \(\X\) over \((\zo^n)^N\) is an \emph{\((N,K,n,k)\)-adversarial source} if it is of the form \(\X=(\X_1,\X_2,\dots,\X_N)\), where each \(\X_i\) is an independent source over \(\zo^n\), and \emph{at least \(K\) of them are good}: i.e., there is some set \(S\subseteq[N]\) of size \(K\) such that \(H_\infty(\X_i)\geq k\), for all \(i\in S\).
\end{definition}

In fact, the authors in \cite{adversarial-sources} provide a more general definition that also allows for some limited dependence between the sources, but \cref{def:intro:adversarial-sources} is already general enough to capture many of their motivating applications and generalize several well-studied settings: \((N,N,n,k)\)-adversarial sources capture the \emph{independent source} model \cite{chor1988unbiased}, \((N,K,1,1)\)-adversarial sources capture \emph{oblivious bit-fixing sources} \cite{chor1985bit}, and \((N,K,n,n)\)-adversarial sources capture so-called \emph{symbol-fixing sources} \cite{kamp2006deterministic}. 

In this section, we will show how to combine our designs from \cref{sec:designs} with a specific kind \emph{leakage-resilient extractor} (LRE) known as \emph{extractors for cylinder intersection} that was introduced in \cite{kms} (see \cref{def:lre}), in order to obtain improved extractors for adversarial sources. The following is our main result of the section:

\begin{theorem}[\cref{thm:main:adversarial}, restated]\label{thm:restated:adversarial}
 There is a universal constant \(C>0\) such that for any fixed \(\delta>0\) and all sufficiently large \(N,K,n,k\in\N\) satisfying \(k\geq\log^Cn\) and \(K\geq N^\delta\), there exists an explicit extractor \(\Ext:(\zo^n)^N\to\zo^m\) for \((N,K,n,k)\)-adversarial sources, with output length \(m=k^{\Omega(1)}\) and error~\(\epsilon=2^{-k^{\Omega(1)}}\).
\end{theorem}

Previously, the best explicit extractor for this setting was constructed by Chattopadhyay et al. \cite{adversarial-sources}, and required \(K\geq N^{0.5+o(1)}\) good sources. On the other hand, it is easy to give a \emph{non-explicit} extractor that requires just \(K\geq2\) good sources.\footnote{This extractor calls an optimal two-source extractor over every pair of sources in the adversarial source, and takes the XOR of the results\cite{chattopadhyay2016extractors}. To see why this works, we refer the reader to a similar proof sketch for a slightly more involved construction, provided in the following paragraphs.} Thus, while our explicit constructions greatly improve the state-of-art (and most notably break the ``\(\sqrt{N}\) barrier''), there is still a lot of room for improvement. Further improvement, however, will require significantly new techniques.

In order to prove \cref{thm:restated:adversarial}, we start by reviewing the \emph{activation vs. fragile correlation} paradigm from \cite{adversarial-sources} for extracting from adversarial sources in \Cref{sec:par_adv_sou}. We prove \Cref{thm:restated:adversarial} in \Cref{sec:new_frm_lke_ed}, where  we will describe how to extend   the  activation vs. fragile correlation technique into a  general framework for extracting from adversarial sources. We use this new framework by combining the recent explicit LREs from \cite{focs2020merged} with our new explicit designs from \cref{sec:designs} to obtain our adversarial source extrators.

\subsection{The \emph{activation vs. fragile correlation} paradigm of \cite{adversarial-sources}} \label{sec:par_adv_sou}
Our construction leverages the ``\emph{activation vs. fragile correlation}'' paradigm introduced in \cite{adversarial-sources} for extracting from adversarial sources. This paradigm was first introduced in an attempt to construct a low-error extractor for \((N,K,n,k)\)-adversarial sources, given just \(k\geq\polylog n\) entropy and as a few good sources, \(K\), as possible. Since there exists a three-source extractor \(\Ext_0\) for \(k_0\geq\polylog n\) entropy and exponentially small error \cite{li2015three}, a natural idea is to somehow employ this object as a subroutine. Using this idea, \cite{adversarial-sources} proposed an extractor for adversarial sources that works as follows. Given as input an adversarial source \(\X=(\X_1,\dots,\X_N)\), the extractor carefully selecting triples of sources, calls \(\Ext_0\) over each triple, and XORs the results. \cite{adversarial-sources} argued that this procedure outputs uniform bits as long as the following two properties hold:

\begin{enumerate}
\item \textbf{Activation}: some \(\Ext_0\) call is \emph{activated}, i.e., only given good sources as input.
\item \textbf{Fragile correlation}: fixing the (XOR of the) output of all other \(\mathsf{Ext}_0\) calls does not affect the output of the activated \(\mathsf{Ext}_0\) call (with high probability).
\end{enumerate}

It is not hard to see why these conditions suffice: \emph{activation} guarantees that some \(\Ext_0\) call outputs uniform bits, while \emph{fragile correlation} guarantees that these uniform bits will be propagated through to the overall output of the extractor (by \cref{fact:stat-dist:add-constant} and \cref{fact:stat-dist:convex-combination}). Thus, the main challenge considered in \cite{adversarial-sources} is determining \emph{how to select triples} such that activation and fragile correlation are guaranteed.

The key idea in \cite{adversarial-sources} is to select triples using the hyperedges of a 3-uniform hypergraph, \(G=(V,E)\). Then, we know that activation is guaranteed as long as the good sources \emph{cover} some hyperedge \(e\in E\), which is guaranteed to happen whenever \(K>\alpha(G)\). In order to ensure fragile correlation, \cite{adversarial-sources} observed that it suffices to require that \(G=(V,E)\) is a \emph{partial Steiner triple system}, also known as an \((N,3,2)\)-design. Such a hypergraph guarantees that each \(\Ext_0\) call shares at most one source with the activated \(\Ext_0\) call. Thus, if we start by fixing all sources that are \emph{not} inputs to the activated \(\Ext_0\) call, it is then easy to fix the outputs of all other \(\Ext_0\) calls without introducing correlation between the inputs to the activated call. Furthermore, by \cref{lem:entropy-drop}, we can show that this process barely decreases the entropy of the inputs to the activated call, and thus its output remains uniform.

This shows that the construction above provides a low-error extractor for \((N,K,n,k)\)-adversarial sources, where \(k\geq\polylog n\) and \(K>\alpha(G)\). Thus, the goal becomes to explicitly construct an \((N,3,2)\)-design \(G=(V,E)\) with small independence number. Using cap set bounds, Chattopadhyay et al. \cite{adversarial-sources} construct such an object with \(\alpha(G)<N^{0.923}\), and thus gave an explicit extractor when there are \(K\geq N^{0.923}\) good sources. In order to improve this requirement on \(K\), it is natural to try to construct an \((N,3,2)\)-design with smaller independence number. However, this seems difficult, and furthermore the tightness of \cref{thm:rodl-sinajova} implies that this technique cannot possibly give an extractor that requires fewer than \(K\geq N^{0.5+o(1)}\) good sources.

Chattopadhyay et al. \cite{adversarial-sources} take a different approach. By using objects known as \emph{strong two-source condensers} \cite{bcdts18} and \emph{non-malleable extractors} \cite{chattopadhyay2020nonmalleable}, the authors are able to create more robust versions of three-source extractors. These robust extractors have stronger conditioning properties, and allow the authors to use different hypergraphs (beyond \((N,3,2)\)-designs) in their construction. As a result, they are able to reduce the requirement on good sources from \(K\geq N^{0.923}\) to \(K\geq N^{0.5+o(1)}\). Unfortunately, however, the conditioning properties of their robust subroutine extractors are extremely specific, and as a result they can only be combined with very specialized types of hypergraphs. These hypergraphs offer no clean generalization of \((N,3,2)\)-designs, and furthermore they appear to be too specialized to offer any further improvement on \(K\) (and, in particular, break the ``\(\sqrt{N}\) barrier'').

\subsection{A new framework using leakage-resilient extractors and extremal designs}\label{sec:new_frm_lke_ed}
If one hopes to significantly improve \(K\), it appears that one would need a multi-source extractor with \emph{even stronger} conditioning properties to use as a subroutine. Recently, exactly such an object was constructed in \cite{focs2020merged}, and is known as a \emph{leakage-resilient extractor} (LRE). LREs are very general objects with extremely strong conditioning properties. The exact variant that will be useful here is actually a specialization known as \emph{extractors for cylinder intersections}, first introduced in \cite{kms}. Informally, we define an \emph{\((r,s)\)-leakage-resilient extractor} to be an \(r\)-source extractor \(\mathsf{LRE}\) that outputs bits that look uniform, \emph{even conditioned on} the output of several functions that each act on \emph{fewer than \(s\)} of the inputs to \(\mathsf{LRE}\). Formally, it is defined as follows.

\begin{definition}[\hspace{1sp}\cite{kms,focs2020merged}]\label{def:lre}
A function \(\LRE:(\zo^n)^r\to\zo^m\) is an \emph{\((r,s)\)-leakage-resilient extractor} for entropy \(k\) and error \(\epsilon\) if the following holds. Let \(\X:=(\X_1,\dots,\X_r)\) be any \(r\) independent \((n,k)\) sources, let \(\mathcal{T}:={[N]\choose s-1}\), and let \(\mathcal{L}:=\{\Leak_T:(\zo^n)^{s-1}\to\zo^m\}_{T\in\mathcal{T}}\) be any collection of functions. Then:
\[
|\LRE(\X)\circ(\Leak_S(\X_S))_{S\in\mathcal{S}}-\U_m\circ(\Leak_S(\X_S))_{S\in\mathcal{S}}|\leq\epsilon.
\]
\end{definition}

Given such a robust extractor, it is now easy to generalize the original extractor of \cite{adversarial-sources} in a clean, natural way: instead of calling a three-source extractor over the hyperedges of an \((N,3,2)\)-design and XORing the results, we call an \((r,s)\)-leakage-resilient extractor over the hyperedges of an \((N,r,s)\)-design and XOR the results. Once again, we can ensure \emph{activation} as long as the number of good sources, \(K\), exceeds the independence number of the design. On the other hand, instead of using \cref{lem:entropy-drop} to ensure fragile correlation, we simply use the leakage-resilience of our leakage-resilient extractor: to see why this works, simply observe that an \((N,r,s)\)-design guarantees that the intersection of two hyperedges has size \(<s\), while a leakage-resilient extractor guarantees to output uniform bits \emph{even conditioned on} several leaks that each act on \(<s\) of its inputs.

Formally, we prove the following lemma, which provides a framework for combining leakage-resilient extractors with general designs in order to extract from adversarial sources.

\begin{lemma}\label{lem:LRE-design-framework}
Let \(G=([N],E)\) be an \((N,r,s)\)-design with independence number \(\alpha\), and let \(\Ext_0:(\zo^n)^r\to\zo^m\) be an \((r,s)\)-leakage resilient extractor for entropy \(k_0\) with error \(\epsilon_0\). Then for any \(K>\alpha\) and \(k\geq k_0\), the function \(\Ext_G:(\zo^n)^N\to\zo^m\) defined as
\[
\Ext_G(X) := \bigoplus_{e\in E(G)}\Ext_0(X_e)
\]
is an extractor for \((N,K,n,k)\) adversarial sources with error \(\epsilon=\epsilon_0\).
\end{lemma}
\begin{proof}
Let \(\X\) be an \((N,K,n,k)\) adversarial source. We must show that \(|\Ext_G(\X)-\U_m|\leq\epsilon\). Because \(K>\alpha\), there is some \(e^\ast\in E\) containing only good sources, i.e., \(\X_i\) has entropy at least \(k\) for each \(i\in e^\ast\). Without loss of generality, we assume \(e^\ast=[r]\). We now fix all other sources \(\mathbf{Z}_1=(\X_j)_{j\notin e^\ast}\), using \cref{fact:stat-dist:convex-combination}:
\[
|\Ext_G(\X)-\U_m|\leq\E_{z_1\sim\mathbf{Z}_1}[|(\Ext_G(\X)\mid\mathbf{Z}_1=z_1)-\U_m|].
\]

Consider any \(z_1=(x_j)_{j\notin e^\ast}\). For each \(e\in E(G)\), we define the restriction \(\Ext_0^e : (\zo^n)^{|e\cap e^\ast|}\to\zo^m\) as \(\Ext_0^e(Y_1,\dots,Y_{|e\cap e^\ast|})=\Ext_0(Y_1,\dots,Y_{|e\cap e^\ast|},(x_j)_{j\in e\setminus e^\ast})\), so that we may write
\[
(\Ext_G(\X)\mid\mathbf{Z}=z_1)=\bigoplus_{e\in E(G)}\Ext^e_0(\X_{e\cap e^\ast})=\Ext_0(\X_{e^\ast})\oplus\bigoplus_{e\in E(G)\setminus\{e^\ast\}}\Ext_0^e(\X_{e\cap e^\ast}).
\]

Because \(G\) is an \((N,r,s)\)-design, any two edges share at most \(s-1\) vertices. Thus, we may partition \(E(G)\setminus\{e^\ast\}\) into \({r\choose s-1}\) sets, depending on the intersection behavior of each edge with \(e^\ast\). In particular, for each \(S\in{e^\ast\choose s-1}\), we define:
\[
\mathcal{W}_S := \{e\in E : e\cap e^\ast\subseteq S\}.
\]
If any \(e\in E\) ends up in more than one \(\mathcal{W}_S\), we simply remove it from all but one of these sets. Now, for each \(S\in{e^\ast\choose s-1}\), we define \(\Leak_S:(\zo^n)^{s-1}\to\zo^m\) such that for any \(X\in(\zo^n)^N\), \(\Leak_S(X_S)=\bigoplus_{e\in\mathcal{W}_S}\Ext_0^e(X_{e\cap e^\ast})\), which is a valid definition because \(e\cap e^\ast\) is always in \(S\), by definition of \(\mathcal{W}_S\). We may now write
\begin{equation}\label{eq:adversarial-rewrite}
    (\Ext_G(\X)\mid\mathbf{Z}_1=z_1)=\Ext_0(\X_{e^\ast})\oplus\bigoplus_{S\in{e^\ast\choose s-1}}\Leak_S(\X_S).
\end{equation}
To bound the distance of this random variable from uniform, we now define the second random variable we will fix, \(\mathbf{Z}_2:=(\Leak_S(\X_S))_{S\in{e^\ast\choose s-1}}\). Fixing this random variable, we have:
\begin{align*}
|(\Ext_G(\X)\mid\mathbf{Z}_1=z_1)-\U_m|&\leq\E_{z_2\sim\mathbf{Z}_2}[|(\Ext_G(\X)\mid\mathbf{Z}_1=z_1,\mathbf{Z}_2=z_2)-\U_m|]\\
&=\E_{z_2\sim\mathbf{Z}_2}[|(\Ext_0(\X_{e^\ast})\mid\mathbf{Z}_2=z_2)-\U_m|]\\
&=|\Ext_0(\X_{e^\ast})\circ\mathbf{Z}_2-\U_m\circ\mathbf{Z}_2|,
\end{align*}
where the first and last (in)equalities follow easily from the definition of statistical distance, and the second (in)equality follows from \cref{eq:adversarial-rewrite} and the fact that adding a constant to a random variable does not change its distance from uniform. But notice that by definition of \(\mathbf{Z}_2\) and the leakage-resilience of \(\Ext_0\), this quantity is bounded above by \(\epsilon_0\), which completes the proof.
\end{proof}

In order to highlight the generality of this framework, we observe that by \cref{lem:entropy-drop}, a standard three-source extractor is, in fact, a \((3,2)\)-leakage-resilient extractor (up to some negligible loss in parameters). Thus, by instantiating \cref{lem:LRE-design-framework} with \(r=3,s=2\), we recover the original extractor and analysis of \cite{adversarial-sources}. Even better, since \cref{thm:rodl-sinajova} tells us that the independence number \(\alpha\) of an \((N,r,s)\)-design decreases quickly as \(r,s\) grow large together, we see that \cref{lem:LRE-design-framework} offers a concrete way to construct extractors for adversarial sources with much fewer good sources, \(K\).

If we want to realize the above plan, we need two explicit objects. First, we need explicit \((N,r,s)\)-designs with independence numbers that decrease quickly as \(r,s\) grow together. \cref{thm:main:designs} of the current paper gives exactly this, and in fact the independence numbers of our designs decrease with \(r,s\) \emph{almost as quickly as possible}, as shown by the tightness of \cref{thm:rodl-sinajova}.

Second, we need explicit leakage-resilient extractors for polylogarithmic entropy that have exponentially small error. Very recently, these exact objects were constructed:

\begin{theorem}[\hspace{1sp}\cite{focs2020merged}]\label{thm:imported:LRE}
There is a universal constant \(C>0\) such that for any sufficiently large \emph{constant} \(r\in\N\) and all \(n,k\in\N\) satisfying \(k\geq\log^C n\), there exists an explicit \((r,r-1)\)-leakage resilient extractor \(\Ext:(\zo^n)^r\to\zo^m\) for min-entropy \(k\) with output length \(m=k^{\Omega(1)}\) and error \(\epsilon=2^{-k^{\Omega(1)}}\).
\end{theorem}

By combining these explicit LREs with our explicit designs, we can finally prove \cref{thm:restated:adversarial}, which significantly improves the adversarial source extractors of \cite{adversarial-sources}.

\begin{proof}[Proof of \cref{thm:restated:adversarial}]
Let \(C\) be the same universal constant from \cref{thm:imported:LRE}, and let \(r\in\N\) be a sufficiently large (even) constant such that \(2/r<\delta\), and such that \cref{thm:imported:LRE} guarantees the existence of an explicit \((r,r-1)\)-leakage resilient extractor \(\Ext_0:(\zo^n)^r\to\zo^m\) for min-entropy \(k\geq\log^Cn\) with output length \(m=k^{\Omega(1)}\) and error \(\epsilon=2^{-k^{\Omega(1)}}\). For sufficiently large \(N\in\N\), \cref{thm:main:designs} guarantees the existence of an \((N,r,r-1)\)-design \(G\) with independence number \(\alpha<N^{\delta}\) that is computable in \(\poly({N\choose r})=\poly(N)\) time. The result now follows by instantiating \cref{lem:LRE-design-framework} with \(\Ext_0\) and \(G\).
\end{proof}

Next, we will show how to use our new and improved low-error extractors for adversarial sources (\cref{thm:restated:adversarial}) to obtain improved improved low-error extractors for small-space sources.

\section{A reduction from small-space sources to adversarial sources}\label{sec:small-space}
In this section, we will show how to use our extractors from \cref{sec:adversarial-sources} to obtain better extractors for small-space sources (as defined by \cref{def:small-space-source}). We will prove the following.

\begin{theorem}[\cref{thm:main:space}, restated]\label{thm:restated:space}
For any fixed \(\delta\in(0,1/2]\) there is a constant \(C>0\) such that for all \(n,k,s\in\N\) satisfying \(k\geq Cn^{1/2+\delta}s^{1/2-\delta}\), there exists an explicit extractor \(\Ext:\zo^n\to\zo^m\) for space \(s\) sources of min-entropy \(k\), with output length \(m=n^{\Omega(1)}\) and error \(\epsilon=2^{-n^{\Omega(1)}}\).
\end{theorem}

Until very recently, the best explicit extractor for this setting \cite{kamp2006small} required entropy \(k\geq Cn^{1-\gamma}s^{\gamma}\), where \(\gamma>0\) is a tiny constant and \(C\) is a large one. In \cite{adversarial-sources}, this requirement was significantly improved to \(k\geq Cn^{2/3+\delta}s^{1/3-\delta}\), for an arbitrarily small constant \(\delta>0\), and the current paper (\cref{thm:restated:space}) further improves this to \(k\geq Cn^{1/2+\delta}s^{1/2-\delta}\). Note that this line of improvements is strict, since we always have \(s<n\) (or else the bounds become trivial). In particular, for say \(s=n^\delta\) space, the entropy requirement has dropped from \(k\geq O(n^{1-\gamma})\) to \(k\geq O(n^{2/3+\delta})\) to \(k\geq O(n^{1/2+\delta})\).

Non-constructively, it is known \cite{kamp2006small} that there exist extractors for space \(s\) sources that have error \(\epsilon\) for min-entropy \(k\geq O(s+\log n+\log(1/\epsilon))\). Thus, while \cref{thm:restated:space} significantly improves the state-of-art in low-error extraction, there is still a lot of room for improvement. However, we note (in \cref{rem:root-n-barrier}) that any substantial improvements to our low-error extractors (i.e., beyond entropy requirement \(k\geq\sqrt{n}\)) will require a new type of reduction that bypasses the need for so-called \emph{total-entropy sources}, which are used in \cite{kamp2006small,adversarial-sources} and are used here. We will see exactly such a reduction in \cref{sec:space:polylog-entropy}. (It will allow us to obtain near-optimal extractors with \emph{polynomial error}. To obtain improved low-error extractors using this new reduction, one needs improved low-error affine extractors.)


We now proceed to prove \cref{thm:restated:space}. The techniques that follow, which will reduce the task of extracting from small space sources to the task of extracting from adversarial sources, are just slightly optimized versions of the exact arguments that appear in \cite{kamp2006small,adversarial-sources}. However, we include them here for completeness. The first step is to reduce small-space sources to a class of sources known as \emph{total entropy sources}, defined as follows.

\begin{definition}\label{sec:end_tot_def}
A random variable \(\X\) over \((\zo^\ell)^r\) is an \emph{\((r,\ell,k)\)-total entropy source} if \(\X=(\X_1,\X_2,\dots,\X_r)\), where each \(\X_i\) is an independent source over \(\zo^\ell\), and \(\sum_{i\in[r]}H_\infty(\X_i)\geq k\).
\end{definition}

In \cite{kamp2006small}, Kamp et al. showed that upon fixing a few positions in the random walk that generates a small space source \(\X\), it is straightforward to use \cref{lem:entropy-drop} to show that \(\X\) becomes a total-entropy source, with high probability. We include the proof for completeness.

\begin{lemma}[\hspace{1sp}\cite{kamp2006small}]\label{lem:small-space-to-TEI}
Let \(\X\) be a space \(s\) source over \(\zo^n\) with min-entropy \(k\). Then for any \(\alpha\in(0,1/4]\) such that \(r=\alpha k/s\) and \(\ell=ns/(\alpha k)\) are positive integers, it holds that \(\X\) is \(2^{-k/4}\)-close to a convex combination of \((r,\ell,k/2)\)-total entropy sources.
\end{lemma}
\begin{proof}
For each \(i\in[n]\), let \(\mathbf{W}_i\sim\zo^s\) be the random variable denoting the state reached in layer \(i\) of the branching program in the random walk that generates \(\X\). Observe that fixing any \(\mathbf{W}_i\) breaks \(\X\) into two independent sources. More generally, observe that if we define \(\mathbf{W}^\ast := (\mathbf{W}_{i\ell+1})_{i\in[0,r-1]}\), then if we condition \(\X\) on any fixing of \(\mathbf{W}^\ast\), it must hold that \(\X\) becomes an \((r,\ell,\Gamma)\)-total entropy source, for \emph{some} \(\Gamma\). Furthermore, by \cref{lem:entropy-drop}, we know
\begin{equation}\label{eq:apply-entropy-drop}
\Pr_{w\sim\mathbf{W}^\ast}[H_\infty(\X\mid\mathbf{W}^\ast=w)\geq k-rs-k/4=k-\alpha k - k/4\geq k/2]\geq 1-2^{-k/4}.
\end{equation}
Thus, the random variable \((\X\mid\mathbf{W}^\ast=w)\) is an \((r,\ell,k/2)\)-total entropy source with probability at least \(1-2^{-k/4}\) over \(w\sim\mathbf{W}^\ast\), which completes the proof.
\end{proof}

The next step is to show that a total-entropy source looks like an adversarial source, using a standard Markov-type argument:

\begin{lemma}\label{lem:TEI-to-adversarial}
Let \(\X\) be an \((r,\ell,\Gamma)\)-total entropy source. Then for any \(N,K,n,k\in\N\) with \(Nn=r\ell\) and \(n\) a multiple of \(\ell\), \(\X\) is also an \((N,K,n,k)\)-adversarial source, as long as \(Kn + Nk\leq\Gamma\).
\end{lemma}
\begin{proof}
By definition of total-entropy source, \(\X=(\X_1,\dots,\X_r)\), where each \(\X_i\) is an independent source over \(\zo^\ell\). By collecting the sources \(\X_i\) into \(N\) buckets containing \(n/\ell\) sources each, we see that \(\X\) is also an \((N,n,\Gamma)\)-total entropy source, and may be rewritten as \(\X=(\X_1,\dots,\X_N)\), where each \(\X_i\) is an independent source over \(\zo^n\). If \(\X\) were not an \((N,K,n,k)\)-adversarial source, then the \(K-1\) highest entropy sources in \(\X\) each have entropy at most \(n\), and the remaining each have entropy \(<k\). This yields \(H_\infty(\X)=\Gamma<(K-1)n+(N-(K-1))k<Kn+Nk\), contradicting the given lower bound on \(\Gamma\).
\end{proof}

Given the above reduction, we can now use our improved adversarial source extractors (from \cref{thm:main:adversarial}) to give improved extractors for total-entropy sources.

\begin{theorem}\label{thm:total-entropy-extractor}
For any fixed \(\delta>0\) and all sufficiently large \(r,\ell,\Gamma\in\N\) with \(\Gamma\geq\max\left\{(r\ell)^{1/2+\delta},r^\delta\ell\right\}\), there exists an explicit extractor \(\Ext:(\zo^\ell)^r\to\zo^m\) for \((r,\ell,\Gamma)\)-total entropy sources, with output length \(m=(r\ell)^{\Omega(1)}\) and error \(\epsilon=2^{-(r\ell)^{\Omega(1)}}\).
\end{theorem}
\begin{proof}
Fix any \(N,n\in\N\) such that \(Nn=r\ell\) and \(n\) is a multiple of \(\ell\). By \cref{lem:TEI-to-adversarial}, every \((r,\ell,\Gamma)\)-total entropy source is also an \((N,K,n,k)\)-adversarial source, provided \(Kn+Nk\leq\Gamma\). Thus, by \cref{thm:main:adversarial}, for any fixed \(\delta_0>0\) there exists an explicit extractor \(\Ext_0:(\zo^r)^\ell\to\zo^{m}\) for \((r,\ell,\Gamma)\)-total entropy sources with output length \(m=n^{\Omega(1)}\) and error \(\epsilon=2^{-n^{\Omega(1)}}\), provided \(N^{\delta_0}n+Nn^{\delta_0}\leq\Gamma\) and \(N,n\) are sufficiently large. To achieve the parameters claimed in the theorem statement, pick \(\delta_0=\delta/2\) and set \(N,n\) as follows: (i) if \(r\geq\ell\), set \(N=n=\sqrt{r\ell}\); (ii) if \(r<\ell\), set \(N=r\) and \(n=\ell\). We conclude by remarking that this casework was motivated by trying to minimize the requirement on \(\Gamma\) by setting \(N=n\). This is not possible in case (ii), but is possible in case (i) by assuming, without loss of generality, that \(r=x^2\ell\) for some \(x\in\N\).
\end{proof}

Previously, the best low-error explicit extractors for total-entropy sources \cite{adversarial-sources} required entropy \(\Gamma\geq\max\{(r\ell)^{2/3+\delta},r^{1/2+\delta}\ell\}\). Non-constructively, we know it is possible \cite{kamp2006small} to achieve an entropy requirement of \(\Gamma\geq O(\ell + \log r)\) and error of \(2^{-\Omega(\Gamma)}\). Thus, while there is still a lot of room to give improved explicit extractors for total-entropy sources, we remark that our total-entropy extractor is almost optimal when the source consists of ``a few long sources'':
\begin{remark}\label{rem:total-entropy-almost-optimal}
The entropy requirement in \cref{thm:total-entropy-extractor} becomes \(k\geq\ell^{1+\delta}\) when \(\ell\geq r\), which is close to the optimal requirement of \(k\geq O(\ell)\).
\end{remark}

Finally, we show how to combine our improved explicit extractors for total-entropy sources (\cref{thm:total-entropy-extractor}) with the standard reduction from small-space sources to total-entropy sources (\cref{lem:small-space-to-TEI}) to complete the proof of \cref{thm:restated:space}:

\begin{proof}[Proof of \cref{thm:restated:space}]
Fix any \(\delta\in(0,1/2]\), and let \(\alpha\in(0,1/4]\) be a sufficiently small constant and \(C>0\) a sufficiently large constant. Given a space \(s\) source \(\X\) over \(\zo^n\) with entropy \(k\geq Cn^{1/2+\delta}s^{1/2-\delta}\), we know by \cref{lem:small-space-to-TEI} that \(\X\) is \(\epsilon_0=2^{-k/4}\)-close to a convex combination of \((r,\ell,k/2)\)-total entropy sources, where \(r=\alpha k/s\) and \(\ell=ns/(\alpha k)\). (Here we assume \(r,\ell\in\N\), but it is easy to extend the argument when this is not the case.) In particular, this means there is some random variable \(\Y\) such that with probability at least \(1-\epsilon_0\) over \(y\sim\Y\), the random variable \((\X\mid\Y=y)\) is an \((r,\ell,k/2)\)-total entropy source.

Let \(\Ext_0:(\zo^\ell)^r\to\zo^m\) be the extractor from \cref{thm:total-entropy-extractor} for such total-entropy sources. We will argue that it also an extractor for the small-space source \(\X\). Notice we have \(|\Ext_0(\X)-\U_m|\leq\E_{y\sim\Y}[|\Ext_0(\X\mid\Y=y)-\U_m|]\leq \epsilon_0+|\Ext_0(\X^\pr)-\U_m|\), where \(\X^\pr\) is some \((r,\ell,k/2)\)-total entropy source. If we can argue that \(r,\ell,k/2\) are sufficiently large and \(k/2\geq\max\{(r\ell)^{1/2+\delta},r^\delta\ell\}\), then \cref{thm:total-entropy-extractor} tells us that \(|\Ext_0(\X)-\U_m|\leq\epsilon_0+|\Ext_0(\X^\pr)-\U_m|\leq 2^{-k/4}+2^{-(r\ell)^{\Omega(1)}}=2^{-n^{\Omega(1)}}\) and \(m=(r\ell)^{\Omega(1)}=n^{\Omega(1)}\), which would prove the current theorem. We know that \(r,\ell,k/2\) are sufficiently large because \(r=\alpha k/s\geq \alpha C(n/s)^{1/2+\delta}\geq\alpha C\), and \(\ell=ns/(\alpha k)\geq1/\alpha\), and \(k\geq C\), where \(\alpha\) is sufficiently small and \(C\) is sufficiently large. Next, we know \(k/2\geq(r\ell)^{1/2+\delta}=n^{1/2+\delta}\) by the provided lower bound on \(k\). Finally, to show \(k/2\geq r^\delta\ell=(\alpha k/s)^\delta ns/(\alpha k)\), rearrange the inequality to obtain \(k^{2-\delta}\geq 2 \alpha^{\delta-1}s^{1-\delta} n\), plug in the provided lower bound on \(k\) to obtain \((Cn^{1/2+\delta}s^{1/2-\delta})^{2-\delta}\geq2\alpha^{\delta-1}s^{1-\delta}n\), and observe that it therefore suffices to show \((0.5C^{2-\delta}\alpha^{1-\delta})\cdot n^{(1/2+\delta)(2-\delta)-1}\geq s^{1-\delta-(2-\delta)(1/2-\delta)}\), or rather
\[
(0.5C^{2-\delta}\alpha^{1-\delta})\cdot n^{2\delta-\delta/2-\delta^2}\geq s^{2\delta-\delta/2-\delta^2}.
\]
This holds because \(n\geq s\) (otherwise the provided lower bound on \(k\) gives \(k>n\)), because \(2\delta-\delta/2-\delta^2\geq 0\) over \(\delta\in(0,1/2]\), and because \(C\) is sufficiently large.
\end{proof}

We conclude this section with a remark about the \(\sqrt{n}\) ``barrier'' in this reduction technique.

\begin{remark}\label{rem:root-n-barrier}
It is not possible to obtain an entropy requirement of \(k<\sqrt{n}\) using the reduction from small-space sources to total-entropy sources from \cref{lem:small-space-to-TEI}, no matter how \(r,\ell\) are set. This is because \(r\ell=n\) implies either (i) \(\ell\geq\sqrt{n}\), or (ii) \(r\geq\sqrt{n}\). In case (i), all of the entropy could be trapped in a single source of length \(>k\), from which extraction is impossible. In case (ii), the application of \cref{eq:apply-entropy-drop} in \cref{lem:small-space-to-TEI} leaves the source with \(0\) bits of entropy, from which extraction is impossible.
\end{remark}

In the following section, we will give a \emph{new} reduction that allows us to bypass the \(\sqrt{n}\) barrier (for polynomial error). We are able to do this because (like in \cite{chattopadhyay2016extractors}), we reduce to a type of independent sources whose lengths need not be determined ahead of time. Unlike total-entropy sources, this will allow us to recurse whenever we get stuck in a tricky situation like case (i) in \cref{rem:root-n-barrier}.

\section{A reduction from small-space sources to affine sources}\label{sec:space:polylog-entropy}
In this section, we construct extractors for small-space sources that can handle just polylogarithmic entropy in the polynomial error regime, proving \cref{thm:MAIN:small-space-polylog-entropy}.

\begin{theorem}[\cref{thm:MAIN:small-space-polylog-entropy}, restated]\label{thm:technical:high-error-small-space}
There exists a universal constant \(C>0\) such that for all \(n,k,s\in\N\) satisfying \(k\geq s\cdot\log^C(n)\), there exists an explicit extractor \(\Ext:\zo^n\to\zo^m\) for space \(s\) sources with min-entropy \(k\), which has output length \(m=(k/s)^{\Omega(1)}\) and error \(\epsilon=n^{-\Omega(1)}\).
\end{theorem}

The main tool we use to prove this theorem is a new reduction from small-space sources to affine sources. As we have seen, an affine source is simply a uniform distribution over some affine subspace of \(\F_2^n\). It will be useful, however, to have the following formal definition.

\begin{definition}[Affine source]\label{def:affine-source:main}
A distribution \(\X\) over \(\F_2^n\) is an \emph{affine source} with min-entropy \(k\) if there exists some \emph{shift vector} \(v_0\in\F_2^n\) and linearly independent basis vectors \(v_1,v_2,\dots,v_k\in\F_2^n\) such that \(\X\) is generated by sampling \(k\) bits uniformly at random \(\mathbf{x}_1,\mathbf{x}_2,\dots,\mathbf{x}_k\sim\F_2\) and computing \(v_0+\sum_{i\in[k]}\mathbf{x}_iv_i\).
\end{definition}

Given this definition, we are now ready to define the main lemma used in proving \cref{thm:technical:high-error-small-space}.



\begin{lemma}[\cref{thm:MAIN:structural-result}, restated]\label{lem:small-space-to-affine}
Let \(\X\) be a space \(s\) source over \(\zo^n\) with min-entropy \(k\). Then \(\X\) is \(\epsilon\)-close to a convex combination of affine sources with min-entropy \(\Gamma\), where
\[
\Gamma=\Omega\left(\frac{k}{s\log(n/k)}\right),
\]
and \(\epsilon=2^{-\Omega(k)}\).
\end{lemma}

Before proving \cref{lem:small-space-to-affine}, we use it to prove \cref{thm:technical:high-error-small-space}. We recall the standard fact that if an extractor works for each source \(\X\) in a family \(\mathcal{X}\) of distributions, then it also works for any convex combination of sources from that family. In particular, this means that any extractor for affine sources is automatically an extractor for small-space sources, by \cref{lem:small-space-to-affine}. The following affine extractor of Li \cite{li2016improved}, which can handle polylogarithmic entropy, will be of particular interest.

\begin{theorem}[\hspace{1sp}\cite{li2016improved}]\label{thm:affine-extractor-li}
There exists a universal constant \(C>0\) such that for all \(n,k\in\N\) satisfying \(k\geq\log^Cn\), there exists an explicit extractor \(\Ext:\zo^n\to\zo^m\) for affine sources with min-entropy \(k\), which has output length \(m=k^{\Omega(1)}\) and error \(\epsilon=n^{-\Omega(1)}\).
\end{theorem}

Resetting the universal constant \(C\) as necessary, \cref{thm:technical:high-error-small-space} follows immediately by combining \cref{lem:small-space-to-affine} and \cref{thm:affine-extractor-li}. Furthermore, since our reduction (\cref{lem:small-space-to-affine}) has extremely low error, we note that we can also combine it with a classical affine extractor  of Bourgain \cite{bourgain2007construction} to immediately get the following bonus result:

\begin{theorem}\label{thm:super-duper-low-error}
For any fixed constants \(C,\delta>0\) and all \(n,k,s\in\N\) satisfying \(k\geq\delta n\) and \(s\leq C\), there exists an explicit extractor \(\Ext:\zo^n\to\zo^m\) for space \(s\) sources with min-entropy \(k\), which has output length \(m=\Omega(n)\) and error \(\epsilon=2^{-\Omega(n)}\).
\end{theorem}

To the best of our knowledge, this is the only nontrivial small-space extractor that achieves super low error \(\epsilon=2^{-\Omega(n)}\), as all previous constructions \cite{kamp2006small} have error at least \(\epsilon=2^{-\widetilde{\Omega}(n)}\). This improvement in error is extremely minor, but we include it as a nice demonstration that our \emph{reduction} has very low error and thus has the capability to produce low-error small-space extractors; our main application of it (\cref{thm:technical:high-error-small-space}), however, will use a polynomial-error affine extractor, whose error will subsume the very low error of the reduction.

At last, we are ready to prove \cref{lem:small-space-to-affine}, which will immediately yield \cref{thm:technical:high-error-small-space} (and \cref{thm:super-duper-low-error}). We do so in the following subsection.



\subsection{A reduction from small-space sources to simple bit-block sources}
In this subsection, we actually show a stronger result than \cref{lem:small-space-to-affine}. In particular, we prove that the reduction holds even for a special case of affine sources called \emph{bit-block sources}. Given a vector \(v\in\F_2^n\), we define \(\supp(v)\subseteq[n]\) to be the subset of all coordinates where \(v\) takes the value \(1\), and we define these sources as follows:

\begin{definition}[\hspace{1sp}\cite{viola2014extractors}]\label{def:bit-block}
A source \(\X\) over \(\F_2^n\) is a \emph{bit-block source} with min-entropy \(k\) if it is an affine source with min-entropy \(k\) (as per \cref{def:affine-source:main}) with the additional guarantee that \(\supp(v_i)\cap\supp(v_j)=\emptyset\), for all \(i\neq j\in[k]\).
\end{definition}
In fact, we even show that the reduction holds for a special case of bit-block sources, called \emph{simple bit-block sources}.
\begin{definition}\label{def:s-bit-block}
A source \(\X\) over \(\F_2^n\) is a \emph{simple bit-block source} with min-entropy \(k\) if it is a bit-block source with min-entropy \(k\) (as per \cref{def:bit-block}), with the additional guarantee that \(\max(\supp(v_i))<\min(\supp(v_j))\) for all \(i<j\in[k]\).
\end{definition}

Given these definitions, we are now able to state the technical version of \cref{lem:small-space-to-affine}.
\begin{lemma}[\cref{lem:small-space-to-affine}, technical version]\label{lem:technical-version}
Let \(\X\) be a space \(s\) source over \(\zo^n\) with min-entropy \(k\). Then \(\X\) is \(\epsilon\)-close to a convex combination of simple bit-block sources with min-entropy \(\Gamma\), where
\[
\Gamma=\Omega\left(\frac{k}{s\log(n/k)}\right),
\]
and \(\epsilon=2^{-\Omega(k)}\).
\end{lemma}
Before we prove \cref{lem:technical-version}, we remark that the reduction also works in the \emph{reverse direction}, implying that small-space sources and simple bit-block sources are \emph{roughly equivalent}, up to a factor of about \(s\).
\begin{remark}
Using \cref{def:small-space-source,def:s-bit-block}, it is relatively straightforward to show: a simple bit-block source \(\X\) over \(n\) bits with min-entropy \(\Gamma\) is also a space \(s=1\) source over \(n\) bits with min-entropy \(\Gamma\). Combining this with \cref{lem:technical-version}, we see that simple bit-block sources and space \(s\) sources are \emph{roughly equivalent} (in the low-error convex combination sense), up to a factor of \(\widetilde{O}(s)\).
\end{remark}



Now, in order to prove \cref{lem:technical-version}, we will use an intermediate type of source, called an \emph{independent source sequence}, which is a natural generalization of independent sources to allow for uneven (and unknown) length. We will show that small-space sources are (close to) a convex combination of independent source sequences, which are a convex combination of simple bit-block sources. We prove the latter first.

\begin{definition}
A source \(\X\) over \(\zo^n\) is an \emph{\((n,r,k)\)-independent source sequence} if there exist some (unknown) lengths \(\ell_1,\dots,\ell_r\in[n]\) that sum to \(n\) such that \(\X=(\X_1,\dots,\X_r)\), where each \(\X_i\) is an independent \((\ell_i,k)\)-source.
\end{definition}

\begin{lemma}\label{lem:indep-source-seq-reduct}
Let \(\X\) be an \((n,\Gamma,1)\)-independent source sequence. Then \(\X\) is a convex combination of simple bit-block sources with min-entropy \(\Gamma\).
\end{lemma}
\begin{proof}
By definition, \(\X=(\X_1,\dots,\X_\Gamma)\), where each \(\X_i\) is an independent \((\ell_i,1)\)-source, for some \(\ell_i\in[n]\). We use the nice observation from \cite{adversarial-sources} that any \((\ell,1)\)-source \(\mathbf{Z}\) is a convex combination of affine sources with min-entropy exactly \(1\). Recall this observation goes as follows: since \(\mathbf{Z}\) is an \((\ell,1)\)-source, it follows via a standard argument (see, e.g., \cite{vadhan2012pseudorandomness}) that it is a convex combination of \emph{flat sources} with min-entropy exactly \(1\). But any such flat source \(\mathbf{Z}^\pr\) is, by definition, a uniform distribution over two distinct strings \(x,y\in\zo^\ell\), which must differ at \emph{some} coordinate \(i^\ast\in[\ell]\). Thus \(\mathbf{Z}^\pr_{i^\ast}\) is a uniform bit, and it is easy to verify that every other bit \(\mathbf{Z}^\pr_j,j\neq i^\ast\) is either constantly \(0\), constantly \(1\), or equal to exactly \(\mathbf{Z}^\pr_{i^\ast}\) or \(\mathbf{Z}^\pr_{i^\ast}\oplus1\). Using \cref{def:affine-source:main}, it is now straightforward to show this is an affine source with min-entropy \(1\).

Thus, we can write each \(\X_i\) as a convex combination of affine sources over \(\F_2^{\ell_i}\) with min-entropy \(1\). This means \(\X\) is a convex combination of sources of the form \(\X^\pr=(\X_1^\pr,\dots,\X_\Gamma^\pr)\), where each \(\X_i^\pr\) is still independent and has the same length \(\ell_i\) as before, but is now also guaranteed to be sampled by the process \(v_0^{(i)}+\mathbf{b}_iv_1^{(i)}\), where \(v_0^{(i)},v_1^{(i)}\in\F_2^{\ell_i}\) are fixed vectors with \(v_1^{(i)}\neq 0\), and \(\mathbf{b}_i\sim\F_2\) is a uniform bit. Thus, we can show \(\X^\pr\) is a simple bit-block source with min-entropy \(\Gamma\) as follows. First, define \(v_0:=(v_0^{(1)},v_0^{(2)},\dots,v_0^{(\Gamma)})\in\F_2^n\). Next, for each \(i\in[\Gamma]\), define \(v_i:=(\mathbbm{1}[1=i]\cdot v_1^{(1)},\mathbbm{1}[2=i]\cdot v_1^{(2)},\dots,\mathbbm{1}[\Gamma=i]\cdot v_1^{(\Gamma)})\in\F_2^n\), where \(\mathbbm{1}[\cdot]\) is the indicator function. Then it is straightforward to verify that \(\X^\pr\) is sampled by \(v_0+\sum_{i\in[k]}\mathbf{b}_iv_i\), and that the vectors \(v_0,v_1,\dots,v_k\) satisfy \cref{def:s-bit-block}. Thus \(\X^\pr\) is a simple bit-block source with min-entropy \(k\), and \(\X\) is a convex combination of such sources.
\end{proof}

At last, we are ready to prove that small-space sources are close to a convex combination of independent source sequences. By combining the following lemma with \cref{lem:indep-source-seq-reduct}, we immediately get \cref{lem:technical-version}.



\begin{lemma}
Let \(\X\) be a space \(s\) source over \(\zo^n\) with min-entropy \(k\). Then \(\X\) is \(\epsilon\)-close to a convex combination of \((n,\Gamma,1)\)-independent source sequences, where \(\Gamma=\Omega\left(\frac{k}{s\log(n/k)}\right)\) and \(\epsilon=2^{-\Omega(k)}\).
\end{lemma}
\begin{proof}
Let \(\W\) be the random walk over the width \(2^s\), length \(n\) branching program that generates \(\X=(\X_1,\X_2,\dots,\X_n)\), and for each \(i\in[n]\), let \(\L_i\) be the vertex in layer \(i\) that is traversed by \(\W\). In other words, \((\L_1,\L_2,\dots,\L_n)\) is a random variable over \([2^s]^n\) that lists the vertices visited by \(\W\) in order (excluding the start vertex). Furthermore, for any \(1\leq i<j\leq n\), we define the \emph{slice} \(\X^{(i,j)}:=(\X_{i+1},\dots,\X_j)\).

For any indices \(0=i_0<i_1<\dots<i_T=n\), it is straightforward to verify that the slices \(\X^{(i_0,i_1)},\X^{(i_1,i_2)},\dots,\X^{(i_{T-1},i_T)}\) become mutually independent when conditioned on fixing \(\L_{i_1},\dots,\L_{i_{T-1}}\) to any \(\ell_{i_1},\dots,\ell_{i_{T-1}}\). Furthermore, given such a fixing, if we can guarantee that  \(T^\pr\) of these slices still have min-entropy at least \(1\) after this fixing, then the source \(\X\) conditioned on this fixing must be an \((n,T^\pr,1)\)-independent source sequence. This is simply because for each ``good'' slice with min-entropy \(1\), we can just concatenate to it all slices preceding it (until we reach another good slice or the start of the source), and for the last good slice with min-entropy \(1\), we can just concatenate to it all slices following it (until we reach the end of the source). Thus, the goal of this proof is to pick layers \(\L_i\) to fix such that with high probability over these fixings, we can make the abovementioned guarantee for the largest \(T^\pr\) possible. By the law of total probability, this will immediately imply \(\X\) is close to a convex combination of \((n,T^\pr,1)\)-independent source sequences.

Let \(\Gamma,t\) be parameters that we will set later. Informally, we will pick layers to fix in the following manner. We will split up the branching program into \(2\Gamma\) slices, and fix the layers in between them. Then, we will argue that with high probability, \(\X\) still has most of its entropy, and so we must be in one of two situations: (1) the \(\Gamma\) slices with the most entropy out of the \(2\Gamma\) slices each have at least \(1\) bit of entropy; or (2) they do not. In case (1), \(\X\) already looks like an \((n,\Gamma,1)\)-independent source sequence, and we are done. In case (2), we know the entropy must be highly concentrated in the \(\Gamma\) highest entropy slices, and so we can recurse on this sub-source that has half the size as the original source, but much more entropy. We will argue that it is impossible to forever avoid case (1) in this recursion, by showing that otherwise we would eventually (after at most \(t\) steps) find a sub-source with more entropy than its length, a contradiction. We will now describe our fixings more formally.

\paragraph{Fixings} We pick layers to fix as follows.\footnote{To reduce notation, we assume throughout the proof that all divisions yield positive integers. It is easy to extend the arguments to handle when this is not the case.} We start by defining a set of indices \(I^{(0)}\) that \emph{split} the branching program into \(2\Gamma\) slices of the same size. Namely, we define indices \(0=i_0^{(0)}<i_1^{(0)}<\dots<i_{2\Gamma-1}^{(0)}<i_{2\Gamma}^{(0)}=n\) such that \(i_j^{(0)}-i_{j-1}^{(0)}=\frac{n}{2\Gamma}\) for all \(j\in[2\Gamma]\), and set \(I^{(0)}:=\{i_j^{(0)}: j\in[2\Gamma-1]\}\). \emph{We now fix \((\L_i)_{i\in I^{(0)}}\) to some string \(\ell^{(0)}\in[2^s]^{2\Gamma-1}\)}.

In order to decide what to fix next, we construct a set \(B^{(0)}\) of the slices induced by \(I^{(0)}\), and we construct a set \(A^{(0)}\subseteq B^{(0)}\) of the \(\Gamma\) highest entropy slices indexed by \(B^{(0)}\), conditioned on the fixing we just performed. More formally, we define \(B^{(0)}=\{(i_0^{(0)},i_1^{(0)}),(i_1^{(0)},i_2^{(0)}),\dots,(i_{2\Gamma-1}^{(0)},i_{2\Gamma}^{(0)})\}\). We now pick the \(\Gamma\) \emph{largest} elements from \(B^{(0)}\) to create \(A^{(0)}\), where \emph{largest} is defined via the following total order: given \((a,b),(c,d)\in B^{(0)}\), we say \((a,b)>(c,d)\) if \(H_\infty(\X^{(a,b)}\mid(\L_i)_{i\in I^{(0)}}=\ell^{(0)})>H_\infty(\X^{(c,d)}\mid(\L_i)_{i\in I^{(0)}}=\ell^{(0)})\); or if these min-entropies are identical and \(a>c\). We now check the min-entropies of the slices in \(A^{(0)}\). If \(H_\infty(\X^a\mid(\L_i)_{i\in I^{(0)}}=\ell^{(0)})\geq1\) for all \(a\in A^{(0)}\), we stop our fixings here. 

Otherwise, we proceed with more fixings. We initialize a counter \(\tau=1\), and use \((\ast)\) to refer to the current location of this text on this page (i.e., the beginning of a loop that we are creating). Then, we define a set of indices \(I^{(\tau)}\) that \emph{split} each of the good slices from the previous round of fixings. More formally, we define \(I^{(\tau)}:=\{\frac{a_1+a_2}{2}: (a_1,a_2)\in A^{(\tau-1)}\}\). \emph{We now fix \((\mathbf{L}_i)_{i\in I^{(\tau)}}\) to some string \(\ell^{(\tau)}\in[2^s]^\Gamma\)}.

In order to decide what to fix next, we construct a set \(B^{(\tau)}\) of the new slices induced by \(I^{(\tau)}\), and we construct a set \(A^{(\tau)}\subseteq B^{(\tau)}\) of the \(\Gamma\) highest entropy slices indexed by \(B^{(\tau)}\), conditioned on all of the fixings we have performed thus far. More formally, we define \(B^{(\tau)}=\{(a_1,\frac{a_1+a_2}{2}):(a_1,a_2)\in A^{(\tau-1)}\}\cup\{(\frac{a_1+a_2}{2},a_2):(a_1,a_2)\in A^{(\tau-1)}\}\). We now pick the \(\Gamma\) \emph{largest} elements from \(B^{(\tau)}\) to create \(A^{(\tau)}\), where \emph{largest} is defined via the following total order: given \((a,b),(c,d)\in B^{(\tau)}\), we say \((a,b)>(c,d)\) if \(H_\infty(\X^{(a,b)}\mid(\L_i)_{i\in I^{(0)}}=\ell^{(0)},(\L_i)_{i\in I^{(1)}}=\ell^{(1)},\dots,(\L_\tau)_{i\in I^{(\tau)}}=\ell^{(\tau)})>H_\infty(\X^{(c,d)}\mid(\L_i)_{i\in I^{(0)}}=\ell^{(0)},(\L_i)_{i\in I^{(1)}}=\ell^{(1)},\dots,(\L_\tau)_{i\in I^{(\tau)}}=\ell^{(\tau)})\); or if these min-entropy are identical and \(a>c\). We now check the min-entropies of the slices in \(A^{(\tau)}\). If \(H_\infty(\X^{a}\mid(\L_i)_{i\in I^{(0)}}=\ell^{(0)},(\L_i)_{i\in I^{(1)}}=\ell^{(1)},\dots,(\L_\tau)_{i\in I^{(\tau)}}=\ell^{(\tau)})\geq1\) for all \(a\in A\), we stop our fixings here. Also, if \(\tau=t\), we stop our fixings here. Otherwise, we increment \(\tau\gets\tau+1\), and we go back to \((\ast)\).\footnote{In order for this process to be well-defined, we should stop if there is a slice \((a,b)\in A^{(\tau)}\) with \(b-a=1\). We will make sure to set our parameter \(t\) to guarantee this.} This concludes our fixings.

\paragraph{Analysis} For convenience, let \(\mathbf{Q}\) denote a single random variable such that fixing \(\mathbf{Q}\) is equivalent to performing all of the fixings described above. Note that \(\mathbf{Q}\) is a deterministic function of \((\L_1,\dots,\L_n)\), and is of the form \((\L_i)_{i\in I}\), where \(I\) is not a single constant subset of \([n]\), but is chosen adaptively. Furthermore, observe that not all elements in the support of \(\mathbf{Q}\) have the same length (depending on when the fixing of layers stopped); indeed, \(\mathbf{Q}\) is a random variable over \([2^s]^{2\Gamma-1}\cup [2^s]^{2\Gamma-1+\Gamma}\cup\dots\cup[2^s]^{2\Gamma-1+t\Gamma}\). However, notice that for every \(q_1,q_2\in\supp(\mathbf{Q})\), \(q_1\) is cannot be a prefix of \(q_2\); it is therefore straightforward to construct an injection from \(\supp(\mathbf{Q})\to[2^s]^{2\Gamma-1+t\Gamma}\), and so \(|\supp(\mathbf{Q})|\leq 2^{s\cdot((t+2)\Gamma-1)}\leq 2^{(t+2)s\Gamma}\).

Recall that we currently have parameters \(\Gamma,t\) that we said we would fix later. We will add \(\epsilon\) to the parameters that we will fix later. The goal now is to show that with probability \(1-\epsilon\) over fixing \(\mathbf{Q}\) to \(q\), the conditional distribution \((\mathbf{X}\mid\mathbf{Q}=q)\) is an \((n,\Gamma,1)\)-independent source sequence, since this would immediately imply \(\X\) is \(\epsilon\)-close to a convex combination of \((n,\Gamma,1)\)-independent source sequences. We would like to show this holds for the best possible choices of \(\Gamma,\epsilon,t\).

We start by invoking \cref{lem:entropy-drop}, which tells us that with probability at least \(1-\epsilon\) over fixing \(\mathbf{Q}\) to \(q\), we have \(H_\infty(\X\mid\mathbf{Q}=q)\geq k-\log(|\supp(\mathbf{Q})|)-\log(1/\epsilon)\geq k-(t+2)s\Gamma-\log(1/\epsilon)\). Consider now some fixing \(\mathbf{Q}=q\) where this holds. We know that there is some \(\tau^\ast\in\{0,1,\dots,t\}\) such that \(q\in[2^s]^{2\Gamma-1+\tau^\ast\Gamma}\), where \(\tau^\ast\) simply counts the number of times we iterated through the fixing loop from above. Recall that by definition, the fixing \(\mathbf{Q}=q\) refers to the fixings \((\L_i)_{i\in I^{(0)}}=\ell^{(0)},\dots,(\L_i)_{i\in I^{(\tau^\ast)}}=\ell^{(\tau^\ast)}\).

Thus, by definition of our fixing procedure, the source \(\X\mid\mathbf{Q}=q\) is simply the concatenation of the slices \((\X^{(\beta,\beta^\pr)}|\mathbf{Q}=q)\), where \((\beta,\beta^\pr)\) ranges over the set \[(B^{(0)}\setminus A^{(0)})\cup(B^{(1)}\setminus A^{(1)})\cup\dots\cup(B^{(\tau^\ast)}\setminus A^{(\tau^\ast)})\cup A^{\tau^\ast},\] and the concatenation happens in increasing order of \(\beta^\pr\). Also notice that the unions above are in fact disjoint. Furthermore, given our discussion at the very beginning of the proof, we know that these slices are mutually independent, because of the conditioning on the layers separating them. Now, we could be in one of two cases: either \(\tau^\ast<t\), or \(\tau^\ast=t\).

Case (1): \(\tau^\ast<t\). In this case, by definition of our fixing procedure, we know that the \(\Gamma\) distinct slices in \((\X\mid\mathbf{Q}=q)\) that are indexed by \(A^{(\tau^\ast)}\) each have min-entropy at least \(1\). Thus, \((\X\mid\mathbf{Q}=q)\) is a sequence of independent slices, with the guarantee that at least \(\Gamma\) of them have min-entropy at least \(1\). By our discussion at the very beginning of this proof, \((\X\mid\mathbf{Q}=q)\) is an \((n,\Gamma,1)\)-independent source sequence.

Case (2): \(\tau^\ast=t\). In this case, observe that in our fixing procedure, we only proceed from iteration \(j\) to \(j+1\) in the loop if some slice in \(A^{(j)}\) has min-entropy \(<1\), which means that all slices in \(B^{(j)}\setminus A^{(j)}\) have min-entropy \(<1\) (since \(A^{(j)}\) contains the \(\Gamma\) slices with the highest min-entropy out of the \(2\Gamma\) slices in \(B^{(j)}\)). Thus, for every \((\beta,\beta^\pr)\in(B^{(0)}\setminus A^{(0)})\cup(B^{(1)}\setminus A^{(1)})\cup\dots\cup(B^{(\tau^\ast-1)}\setminus A^{(\tau^\ast-1)})\), we know \(H_\infty(\X^{(\beta,\beta^\pr)}\mid\mathbf{Q}=q)<1\). (This was also true in the previous case, but we did not need this observation there.) For all other \((\beta,\beta^\ast)\in(B^{(\tau^\ast)}\setminus A^{(\tau^\ast)})\cup A^{(\tau)}\), it trivially holds that \(H_\infty(\X^{(\beta,\beta^\pr)}\mid\mathbf{Q}=q)\leq\beta^\pr-\beta\), since this slice is just a random variable over \(\beta^\pr-\beta\) bits. It is straightforward to show that \(\beta^\pr-\beta=\frac{n}{2\Gamma}\cdot2^{-\tau^\ast}\), since our first slices \(B^{(0)}\) divide the \(n\) bit source into \(2\Gamma\) equal sized pieces, and \(A^{(0)}\subseteq B^{(0)}\), and at each iteration \(j\) of the loop we cut each slice from \(A^{(j-1)}\) in half to get \(B^{(j)}\). Thus, \(H_\infty(\X^{(\beta,\beta^\pr)}\mid\mathbf{Q}=q)\leq\frac{n}{2\Gamma}\cdot2^{-\tau^\ast}\) for all \((\beta,\beta^\ast)\in(B^{(\tau^\ast)}\setminus A^{(\tau^\ast)})\cup A^{(\tau)}\).

Thus, we know an upper bound on the entropy of each slice in \((B^{(0)}\setminus A^{(0)})\cup(B^{(1)}\setminus A^{(1)})\cup\dots\cup(B^{(\tau^\ast-1)}\setminus A^{(\tau^\ast-1)})\cup(B^{(\tau^\ast)}\setminus A^{(\tau^\ast)})\cup A^{\tau^\ast}\). Furthermore, the sets in the union are disjoint, and each set in this union contains \(\Gamma\) distinct slices, which we have already mentioned are mutually independent. Thus, we have:
\begin{align*}
H_\infty(\X\mid\mathbf{Q}=q)&<(1+\tau^\ast-1)\cdot\Gamma\cdot1 + (1+1)\cdot\Gamma\cdot\left(\frac{n}{2\Gamma}\cdot2^{-\tau^\ast}\right)\\
&=\Gamma\tau^\ast + n\cdot2^{-\tau^\ast}\\
&=\Gamma t + n2^{-t}.
\end{align*}
Combining this with the assumption we made about \(q\) near the beginning of our analysis, we have:
\begin{align}\label{eq:entropy-crossover}
k-(t+2)s\Gamma-\log(1/\epsilon)\leq H_\infty(\X\mid\mathbf{Q}=q)<\Gamma t + n2^{-t}.
\end{align}
We finally arrive at our last goal: setting parameters \(\Gamma,t,\epsilon\). We know that for any setting of these parameters that contradicts \cref{eq:entropy-crossover}, Case (2) simply cannot occur. Thus, for any such setting, we know that with probability \(1-\epsilon\) over fixing \(\mathbf{Q}=q\), we have \(H_\infty(\X\mid\mathbf{Q}=q)\geq k-(t+2)s\Gamma-\log(1/\epsilon)\), and this implies Case (1) must occur. In other words, with probability \(1-\epsilon\) over \(q\sim\mathbf{Q}\), we have that \((\X\mid\mathbf{Q}=q)\) is an \((n,\Gamma,1)\)-independent source sequence, which immediately implies that \(\X\) is \(\epsilon\)-close to a convex combination of \((n,\Gamma,1)\)-independent source sequences.

So all that remains is to pick the best possible \(\Gamma,t,\epsilon\) that contradict \cref{eq:entropy-crossover}, and in particular show that our selected \(\Gamma,\epsilon\) matches the claimed parameters in the theorem statement. We only have one minor restriction in our freedom to pick these parameters. We briefly recall the footnote from our fixings procedure, and note that the only requirement we have is that \(t\) is set so the procedure remains valid; namely, so that for every \(\tau\in[t]\) and \((\beta,\beta^\pr)\in B^{(\tau)}\) created by the fixing procedure, \(\beta^\pr-\beta\geq 1\), since this will ensure that we are creating valid slices. Above, we showed that \(\beta-\beta=\frac{n}{2\Gamma}\cdot2^{-\tau}\), and so the only restriction we have is that \(\frac{n}{2\Gamma}\cdot2^{-t}\geq1\).

Thus, to complete the proof, we may pick any \(\Gamma,t,\epsilon\) that satisfy the above restriction, while contradicting \cref{eq:entropy-crossover}. In particular, these parameters just need to satisfy
\begin{align*}
k-(t+2)s\Gamma-\log(1/\epsilon)&\geq\Gamma t + n2^{-t},\text{ and}\\
\frac{n}{2\Gamma}\cdot2^{-t}&\geq 1.
\end{align*}
Combining these, we just require:
\[
\Gamma\leq\min\left\{\frac{k-n2^{-t}-\log(1/\epsilon)}{t\cdot(s-1)+2s},\frac{n}{2^{t+1}}\right\}.
\]
We take \(t:=\log(4n/k)\) and \(\epsilon:=2^{-k/2}\) and \(\Gamma:=\frac{k}{20s\cdot\log(n/k)}\) to complete the proof.
\end{proof}

\section{Future directions}\label{sec:conclusions}
In this paper, we give new constructions of extractors for small-space sources based on (i) a new reduction from small-space sources to affine sources, and (ii) improved extractors for adversarial sources. The new key ingredient we use for our adversarial source extractors is (the first) derandomization of R\"{o}dl and \v{S}inajov\'{a}'s probabilistic designs \cite{sts}, which we combine with recent explicit constructions \cite{kms,focs2020merged} of a certain kind of leakage resilient extractors, known as extractors for cylinder intersections. These constructions demonstrate new applications of these two pseudorandom objects, and it would be interesting to explore whether these objects have further applications in pseudorandomness and complexity.

Beyond the above, the three most natural open problems are as follows.

\begin{problem}\label{prob:space}
\emph{Better low-error extractors for small-space sources:} Reduce the entropy requirement for low-error small-space extraction (\cref{thm:main:space}) so that it is closer to the entropy requirement for polynomial-error small-space extraction (\cref{thm:MAIN:small-space-polylog-entropy}).
\end{problem}

\begin{problem}\label{prob:adversarial}
\emph{Better extractors for adversarial sources:} Improve the requirement on good sources in \cref{thm:main:adversarial} from \(K\geq N^\delta\) to \(K\geq\polylog(N)\), or (less ambitiously) \(K\geq N^{o(1)}\).
\end{problem}

\begin{problem}\label{prob:designs}
\emph{Better explicit designs with small independence number:} Improve the constant in the power of \(n\) of \cref{thm:main:designs} from 2 to 1.99.
\end{problem}

Given our new reduction from small-space extractors to affine sources, a concrete way to approach \cref{prob:space} is to simply pursue the construction of better low-error affine extractors. In particular, solving the affine extraction problem would effectively also ``finish off'' the small-space extraction problem. Meanwhile, \cref{prob:adversarial} can be solved by constructing a leakage-resilient extractor against \emph{number-on-forehead} leakage: that is, an extractor whose output looks uniform even conditioned on joint functions of all but one of its inputs. Finally, it would be interesting to see if \cref{prob:designs} could be answered using more elaborate properties of specific codes (i.e., beyond their distance and dimension).


\bibliographystyle{alpha}
\bibliography{references}

\end{document}